\documentclass[siglarge,screen,12pt]{acmart}
\makeatletter
\renewcommand\@formatdoi[1]{\ignorespaces}
\makeatother

\markleft{}
\makeatletter
\let\@authorsaddresses\@empty
\makeatother

\renewcommand\footnotetextcopyrightpermission[1]{} 
\usepackage{verbatim}
\usepackage{paralist}
\usepackage{amsfonts}
\usepackage{amsmath}
\usepackage{amsthm}
\usepackage{latexsym}
\usepackage{dsfont}
\usepackage{url}
\usepackage{comment}
\usepackage{todonotes}
\usepackage[thinc]{esdiff}
\usepackage[ruled, linesnumbered]{algorithm2e}
\usepackage{algpseudocode}
\usepackage{bbm}
\usepackage{thmtools}
\usepackage{thm-restate}

\allowdisplaybreaks

\newtheorem{theorem}{Theorem}[section]

\newtheorem{lemma}[theorem]{Lemma}

\newcommand{\partdiff}[2]{\frac{\partial {#1}}{\partial {#2}}}

\newcommand{\mixdiff}[3]{\frac{\partial^2 {#1}}{{\partial {#2}}{\partial {#3}}}}
\newcommand{\iprod}[2]{\left\langle {#1}, 
{#2} \right \rangle}

\newcommand{\eps}{\varepsilon}
\newcommand{\E}{\mbox{\bf E}}

\newcommand{\cE}{{\mathcal E}}

\newcommand{\cQ}{{\mathcal Q}}

\newcommand{\RR}{{\mathbb R}}

\newcommand{\NN}{{\mathbb N}}

\newcommand{\vx}{\mathbf{x}}
\newcommand{\vy}{\mathbf{y}}
\newcommand{\xdot}{\dot \vx}
\newcommand{\ydot}{\dot \vy}
\newcommand{\dt}{\mathrm{d}t}

\def\b1{{\bf 1}}
\def\be{{\bf e}}

\def\bx{{\bf x}}
\def\by{{\bf y}}

\newcommand{\prob}[1]{\textnormal{Pr}(#1)}

\title{A Polynomial Lower Bound on Adaptive Complexity of Submodular Maximization}

\author{Wenzheng Li}
\affiliation{%
  \institution{Stanford University}
  \city{Stanford}
  \state{CA}
  \country{USA}
}
\email{wzli@stanford.edu}

\author{Paul Liu}
\affiliation{%
  \institution{Stanford University}
  \city{Stanford}
  \state{CA}
  \country{USA}
}
\email{paul.liu@stanford.edu}

\author{Jan Vondr\'{a}k}
\affiliation{%
  \institution{Stanford University}
  \city{Stanford}
  \state{CA}
  \country{USA}
}
\email{jvondrak@stanford.edu}

\copyrightyear{2020}
\acmYear{2020}

\ccsdesc[500]{Theory of computation~Parallel computing models}
\ccsdesc[300]{Theory of computation~Distributed computing models}
\ccsdesc[300]{Theory of computation~Discrete optimization}
\ccsdesc[300]{Theory of computation~Approximation algorithms analysis}
\ccsdesc[300]{Theory of computation~Mathematical optimization}

\keywords{submodular, adaptive model, optimization, lower bound, symmetry gap, double greedy}

\begin{document}

\begin{abstract}
In large-data applications, it is desirable to design algorithms with a high degree of parallelization.  In the context of submodular optimization, adaptive complexity has become a widely-used measure of an algorithm's ``sequentiality". Algorithms in the adaptive model proceed in rounds, and can issue polynomially many queries to a function $f$ in each round. The queries in each round must be independent, produced by a computation that depends only on query results obtained in previous rounds.

In this work, we examine two fundamental variants of submodular maximization in the adaptive complexity model:  cardinality-constrained monotone maximization, and unconstrained non-mono-tone maximization. Our main result is that an $r$-round algorithm for cardinality-constrained monotone maximization cannot achieve an approximation factor better than $1 - 1/e - \Omega(\min \{ \frac{1}{r}, \frac{\log^2 n}{r^3} \})$, for any $r < n^c$ (where $c>0$ is some constant).
This is the first result showing that the number of rounds must blow up polynomially large as we approach the optimal factor of $1-1/e$.

For the unconstrained non-monotone maximization problem, we show a positive result: For every instance, and every $\delta>0$, either we obtain a $(1/2-\delta)$-approximation in $1$ round, or a $(1/2+\Omega(\delta^2))$-approximation in $O(1/\delta^2)$ rounds. In particular (and in contrast to the cardinality-constrained case), there cannot be an instance where (i) it is impossible to achieve an approximation factor better than $1/2$ regardless of the number of rounds, and (ii) it takes $r$ rounds to achieve a factor of $1/2-O(1/r)$.
\end{abstract}
\maketitle

\section{Introduction}
Let $E$ be a set of size $n$, and $f\,:\,2^E \rightarrow \mathbb{R}_+$ a function satisfying $f(S \cup \{e\})-f(S)\geq f(T \cup \{e\}) - f(T)$ for all $S \subseteq T \subseteq E \setminus \{e\}$. Such a function is called \emph{submodular}. When $f(S) \leq f(T)$ for all $S \subseteq T$, $f$ is called \emph{monotone}. 
Submodular functions capture a notion of \emph{diminishing returns}, where the gain $f(S \cup \{e\})-f(S)$ gets smaller as $S$ gets larger. This notion arises naturally in combinatorial optimization, algorithmic game theory, and machine learning, among other domains (see ~\cite{KG14} and the references contained therein). As such, there has been a wealth of research on submodularity over the past few decades. 

As datasets grow larger in size however, there has been renewed attention examining submodular optimization in several computing models for large-scale data. These models typically assume oracle access to a submodular function $f$ and restrict the computation in certain ways. Such models include streaming~\cite{KMZLK19, BMKK14, MV17}, MapReduce~\cite{BENW16, LV19, MZ15, KMVV15, BENW15, MKSK13}, and more recently the adaptive complexity model~\cite{BS18, BS18b, BRS19b, BBS19a, CK19, FMZ19a, EN19,  CQ19, ENV19b, KMZLK19}. Algorithms in the streaming model examine elements of $E$ one at a time, with the only constraint being limited memory. Algorithms in the MapReduce model partition the dataset among many machines with smaller memory, and run local computations on each machine independently. In both models, algorithms are typically memory-efficient, with algorithms running sequentially once the data to be processed is small enough. However, in many applications oracle queries to $f$ are the dominant computational bottleneck. Thus, long chains of sequential queries drastically slow down an algorithm.
For these applications, the adaptive model of Balkanski and Singer~\cite{BS18} offers more relevant constraints. Algorithms in the adaptive model proceed in \emph{rounds}, and can issue polynomially many queries to $f$ in each round. The queries in each round must be independent, and can be generated by the algorithm based on query results obtained in previous rounds. Informally, the adaptive model measures complexity by the longest chain of sequentially dependent calls to $f$ in the algorithm. Consequently, standard greedy algorithms which examine one element at a time have essentially worst possible adaptive complexity. 

Over the past two years, there has been a burst of work in adaptivity-efficient algorithms for maximizing submodular functions~\cite{BBS19a, CK19, FMZ19a, EN19, BS18b, BRS19b, CQ19, ENV19b, KMZLK19, BS18}. For maximizing a monotone submodular function under a cardinality constraint, several independently developed algorithms are known~\cite{BRS19, CK19, FMZ19a, EN19} which surprisingly all achieve a close-to-optimal $1-1/e-\eps$ approximation ratio using $O(\frac{1}{\epsilon^2} \log n)$ adaptive rounds. Moreover, Breuer et al.~\cite{BBS19a} have developed an $O(\frac{1}{\eps^2} \log n \log^2 \log k)$-adaptive algorithm that outperforms (in practice) all current theoretically state-of-the-art algorithms in practice. In the case of a matroid constraint, the best theoretical results achieve a $(1-1/e-\eps)$-approximation with $O(\frac{1}{\epsilon^3} \log n \log k)$ adaptivity~\cite{BRS19b, CQ19}, where $k$ is the rank of the underlying matroid. 
The known results are somewhat weaker in the cardinality-constrained non-monotone setting. Mirrokni et al. \cite{FMZ19b} developed a $(0.039-\eps)$-approximation in $O(\frac{1}{\eps} \log n)$ adaptive rounds, Balkanski et al. \cite{BBS18} designed a $1/2e$-approximation in $O(\log^2 n)$ rounds, and Ene et al.~\cite{ENV19b} achieved a $(1/e-\eps)$-approximation in $O(\frac{1}{\eps^2} \log^2 n)$ rounds. 

In the unconstrained non-monotone case, two independent works developed a $(1/2-\eps)$-approximation in $\tilde{O}(1 / \eps)$ rounds \cite{CFK19,ENV19}. This is achieved through a low-adaptivity version of the \emph{double greedy} algorithm of Buchbinder et al~\cite{BFNS15}. Interestingly, the unconstrained non-monotone case seems unique in that the number of rounds is {\em independent of $n$}.
From hardness results in the sequential model, it is known that $1/2$ is the optimal factor for unconstrained submodular maximization \cite{FMV11}, and $1-1/e$ is optimal for cardinality constrained monotone submodular maximization~\cite{Feige98}.

Considering that so many different algorithms are exhibiting a similar behavior --- adaptive complexity blowing up polynomially as we approach the optimal approximation factor, a natural question is whether this is necessary. The only non-trivial lower bound that we are aware of appears in the initial work by Balkanski and Singer~\cite{BS18}: $\Omega\left(\frac{\log n }{\log \log n}\right)$ rounds are necessary to achieve a $\frac{1}{\log n}$-approximation for the cardinality-constrained monotone maximization problem. No stronger lower bounds were known for achieving constant factors, even close to $1-1/e$.

\subsection{Our results}

In this work, we prove two main results. 

\paragraph{Monotone submodular maximization.}
Our first result is a {\em polynomial} lower bound on the adaptive complexity of algorithms for cardinality-constrained monotone submodular maximization.

\begingroup
\def\thetheorem{\ref{thm:monotone-log-lb}}
\begin{theorem}
For any $r, n \in \NN$, $r < n^c$, where $c>0$ is some absolute constant, there is no algorithm using $r$ rounds of queries and achieving better than a $\left(1-1/e-\Omega\left(\min \{ \frac{1}{r}, \frac{\log^2 n}{r^3} \}\right)\right)$-approximation for monotone submodular maximization subject to a cardinality constraint (on a ground set of size $n$). 
\end{theorem}
\addtocounter{theorem}{-1}
\endgroup

This is the first result showing that if we want to approach the optimal factor of $1-1/e$, the adaptive complexity must blow up to polynomially large factors.
(The hard instances of \cite{BS18} are unrelated to the factor of $1-1/e$; they allow one to compute the optimal solution in $O\left(\frac{\log n}{\log \log n}\right)$ rounds.)

As the statement of the result suggests, we consider two regimes of $\eps := \Omega(\min \{ \frac{1}{r}, \frac{\log^2 n}{r^3} \})$: For $\eps > c' / \log n$, our lower bound on the number of rounds to achieve a $(1-1/e-\eps)$-approximation is $\Omega(1/\eps)$. For $\eps < c' / \log n$, our lower bound is $\Omega\left(\frac{\log^{2/3} n}{\eps^{1/3}}\right)$. Apart from this quantitative difference, the first regime has another feature: We provide a single (randomized) instance for a given $n$ such that achieving a $(1-1/e-\eps)$-approximation requires $\Omega(1/\eps)$ rounds for every $\eps > c' / \log n$. The instances in the second regime are different depending on the value of $\eps$.

As a building block for our result, we design a simple hard instance for (sequential) monotone submodular maximization. It implies the following result, which could be of independent interest.

\begingroup
\def\thetheorem{\ref{thm:1-1/e}}
\begin{theorem}
For the monotone submodular maximization problem subject to a cardinality constraint, $\max \{f(S): |S| \leq k \}$, any $(1-1/e+\Omega(n^{-1/4}))$-approximation algorithm on instances with $n$ elements would require exponentially many value queries.
\end{theorem}
\addtocounter{theorem}{-1}
\endgroup

As far as we know, the lower bounds known so far \cite{Feige98,V09,V13} only showed the hardness of $(1-1/e+\Omega(1/\log n))$-approximations. 

\paragraph{Unconstrained submodular maximization.}
Following our hardness result, it is natural to ask whether a similar lower bound holds for unconstrained submodular maximization. Here, the optimal approximation factor is $1/2$ and it is known that a $(1/2-\eps)$-approximation can be achieved in $\tilde{O}(1/\eps)$ adaptive rounds~\cite{CFK19, ENV19}. Hence, a lower bound analogous to our first result, where it takes $\Omega(1/\epsilon)$ rounds to approach the optimal factor within $\eps$, would be {\em optimal} here.
Nevertheless, we show that the situation here is substantially different.

\begingroup
\def\thetheorem{\ref{thm:dg-gains}}
\begin{theorem}
Let $f:2^E \rightarrow \RR_+$ be a non-monotone submodular function with maximum value $\textrm{OPT}$, and let $R$ denote a uniformly random subset of $E$. If $\E[f(R)] \leq (1/2 - \delta) OPT$, then there is an algorithm using $O(1/\delta^2)$ adaptive rounds that achieves value $(1/2 + \Omega(\delta^2))OPT$.
\end{theorem}
\addtocounter{theorem}{-1}
\endgroup

In other words, a $1/2-\eps$ approximation to unconstrained non-monotone submodular optimization takes at most $O\left(\min\{\frac{1}{\delta^2}, \frac{1}{\eps}\}\right)$ rounds, where $\E[f(R)] = (\frac12 - \delta) OPT$. Unlike the monotone case, there is no blowup when $\eps \rightarrow 0$, as long as $\E[f(R)]$ is bounded away from $\frac12 OPT$. 

Also, this means that the hardest instances for many rounds are in some sense the easiest ones for 1 round. There are no instances of unconstrained submodular maximization such that (i) it is impossible to achieve a factor better than $1/2$ regardless of the number of rounds, and (ii) it takes $\Omega(1/\eps)$ rounds to achieve a factor of $1/2-\eps$. Either it takes a constant number of rounds to achieve a factor better than $1/2$, or a random set is already very close to $\frac12 OPT$.


\subsection{Our techniques} 

\paragraph{Lower bound for cardinality-constrained maximization.}
Our construction consists of a sequence of ``layers" of elements $X_1,\ldots,X_\ell,Y$ such that one can learn only one layer in one round of queries. At this level, our construction is analogous to that of Balkanski and Singer~\cite{BS18}.  However, a key difference is that the number of layers in \cite{BS18} is limited to $\Theta\left({\log n} / {\log \log n}\right)$ since the construction forces the size of each layer to shrink by a \emph{minimum} factor of $\Theta(\log^c n)$ compared to the previous one. In our construction, the layers shrink by constant factors. Depending on the shrinkage rate of the layers, we obtain a trade-off between the approximation ratio and the adaptive complexity. When the layers shrink by a constant factor, there are $\Theta(\log n)$ layers and the best solution obtained after peeling away $r$ layers is at most $1-1/e-\Theta(1/r)$. To increase the number of layers to $r > \log n$, the shrinkage factor will instead be $1+\Theta\left(\frac{\log n}{r}\right)$ and the lower bound will be $1-1/e - \Omega\left(\frac{\log^2 n}{r^3}\right)$.

The main difficulty is how to design the instance so that (i) the layers shrink by constant factors, and (ii) given $X_1,\ldots,X_r$, we can ``see" only $X_{r+1}$ but not the further layers. In \cite{BS18}, this is achieved by arguing that for non-trivial queries, the number of elements we can possibly catch in $X_{r+2}, X_{r+3}$, etc. is so small that these layers do not make any difference. However, this argument doesn't work when the layers shrink by a constant factor; in this setting it is easy to catch many elements in the layers beyond $r+1$ ($X_{r+2}$, $X_{r+3}$, etc.).


We resolve this issue by an adaption of the {\em symmetry gap} technique: We design a monotone submodular function $f$ as a composition of functions $h_k$ inspired by the symmetry gap construction of \cite{FMV11,V13}; $h_k$ treats a pair of layers $(X_k, X_{k+1})$ in a {\em symmetric way} so that a typical query cannot distinguish the elements of $X_{k}$ from the elements of $X_{k+1}$. 
However, given $X_k$, it is possible to use $h_k$ to determine $X_{k+1}$. This leads to the desired effect of not being able to determine $X_{k+1}$ until one round after we have learned $X_k$.
The indistinguishability argument does not rely on the inability to find many elements in a certain layer, but on the inability to find sets which are significantly {\em asymmetric}, or unbalanced with respect to  different layers.

In addition, we need to design the function so that solutions found after a limited number of rounds are worse than $(1-1/e) OPT$. We achieve this by applying the symmetrization in a way different from previous applications: asymmetric solutions are {\em penalized} here compared to symmetric solutions (except for the optimal solution on $Y$). In addition, an initial penalty is imposed on the set $X_1$, which makes it disadvantageous to pick a uniformly random solution (and this also enables an adaptive algorithm to get started, by distinguishing $X_1$ from the other sets). Finally, $Y$ contains a hard instance showing the optimality of the factor $1-1/e$. In other words, even if we learn all the layers, we still cannot achieve a factor better than $1-1/e$.

This construction can be extended to the setting where the shrinkage factor is closer to $1$,  $|X_{i+1}| = |X_i| /  \left(1 + \frac{\log n}{r}\right)$, with some additional technicalities. The construction in this case is a separate one for each value of $r$; the number of layers is $r+1$ and we argue about the quality of solutions that can be found in $r-1$ rounds. For technical reasons, the hardness factor here approaches $1-1/e$ proportionally to $\log^2 n/r^3$ rather than $1/r$. 

\paragraph{Improved approximation for unconstrained maximization.}
In the non-monotone case, it is known that a random set gets at least $1/4$ of the optimal solution in expectation~\cite{FMV11}. Furthermore, a $1/2$-approximation is best possible \cite{FMV11,V13}, no matter how many rounds of queries are allowed. Thus generally, the value of a random set is expected to be between $\frac14 OPT$ and $\frac12 OPT$. However, for the known cases (e.g. certain bipartite directed cut instances) where a random set achieves exactly $\frac14 OPT$, it is actually easy to find the optimal solution. Conversely, for known hard instances, where double greedy is close to $\frac12 OPT$, a random set also gets roughly $\frac12 OPT$. We prove that there is a trade-off between the performance of the random set and a variant of the double greedy algorithm, by relating them to integrals that can be compared using the Cauchy-Schwarz inequality. Consequently, we prove that the gain of the double greedy algorithm over $\frac12 OPT$ grows at least quadratically in $\delta = \frac12 -  \frac{\E[f(R)]}{OPT}$.

\subsection{Paper organization}
The rest of the paper is organized in the following way. In Section~\ref{sec:log-lower-bound} we give a detailed construction of our lower bound for $r = O(\log n)$ adaptive rounds. In Section~\ref{sec:poly-lower-bound}, we outline a similar construction that extends to $r = O(n^c)$. Section~\ref{sec:non-monotone-alg} presents our improved result for unconstrained submodular maximization. Finally, the appendix contains some basic results as well as a new hardness instance for cardinality-constrained submodular maximization that may be of independent interest.
\section{Log-round lower bound for monotone submodular maximization}
\label{sec:log-lower-bound}

In this section we describe our hard instance which proves the following result.

\begin{theorem}
\label{thm:monotone-log-lb}
For any $n$, there is a family of instances of monotone submodular maximization subject to a cardinality constraint on a ground set of size $n$, such that for any $r < \frac13 \log_2 n$, any algorithm using $r$ rounds of queries achieves at most a $(1-1/e-\Theta(1/r))$-approximation.
\end{theorem}

We note that this construction works only for $r = O(\log n)$; we provide a more general construction later in the paper, which works for $r = O(n^c)$ for some constant $c>0$. The two constructions are in fact quite similar. An advantage of the construction for $r = O(\log n)$, apart from ease of exposition, is that it provides a single instance for which it is hard to achieve better than $(1-1/e-\Theta(1/r))$-approximation ir $r$ rounds for any $r = O(\log n)$, and hard to achieve a better than ($1-1/e)$-approximation in any number of rounds. We discussed in the introduction why this is interesting vis-a-vis the unconstrained optimization problem.

\subsection{Construction of the objective function}

Let $E = X_1 \cup \ldots \cup X_\ell \cup Y_1 \cup \ldots Y_{\ell'}$ where $|E| = 2^{3\ell}$, $|X_i| = 2^{3\ell-i}$ and $k = |Y_i| = \frac{1}{\ell'} 2^{2\ell} = n^{5/8}$, $\ell' = 2^{\ell/8} = n^{1/24}$ ($\ell$ is a multiple of $8$). The exact parameters are not so important. The main requirements are that the $X_i$'s are decreasing by constant factors and that the $Y_i$ are polynomially large (i.e. $\Omega(n^c)$ for some $c < 1$). The partition is uniformly random among all partitions with these parameters. 
We define the objective function as a function of real variables $f(x_1,x_2,\ldots,x_\ell,y_1,\ldots,y_{\ell'})$ where the interpretation of $x_i, y_i$ is that for a set $S \subseteq E$, $x_i = \frac{1}{k} |S \cap X_i|$ and $y_i = \frac{1}{k} |S \cap Y_i|$. 
We consider a function in the following form (where $x_0 = 0$):
$$ \tilde{h}(x_1, \ldots, x_\ell) =  (1-h_1(x_\ell)) \prod_{i=0}^{\ell-1}(1-h(x_i,x_{i+1})), $$
$$ f_1(x_1,\ldots,x_\ell,y_1,\ldots,y_{\ell'}) =1 - \tilde{h}(x_1, \ldots, x_\ell)(1-g(y_1,\ldots,y_{\ell'})).$$
We design $g,h,h_1$ to be ``smooth monotone submodular" functions (with non-negative and non-increasing partial derivatives), which means that the same properties are inherited by $f_1$. The actual objective function is going to be obtained by discretization of the function $f = \min \{ f_1, 1-\eps \}$ ($\eps>0$ to be specified later). This defines a monotone submodular function (we refer the reader to Appendix~\ref{app:basics} for details).

Before proceeding to technical details, we wish to make the following points:
\begin{itemize}
\item $h(x,x')$, $h_1(x)$, and $g(y_1,\ldots,y_{\ell'})$ are going to be non-decreas-ing continuous functions with range $[0,1]$ and non-increasing first partial derivatives. It is well-known that this corresponds to monotone submodular functions in the discrete setting. Also, it is easy to verify that $f(x_1,\ldots,x_\ell,y_1,\ldots,y_{\ell'})$ defined as above inherits the same property of its partial derivatives, and for $x_i = \frac{1}{k} |S \cap X_i|$, $y_j = \frac{1}{k} |S \cap Y_j|$, this defines a monotone submodular function of $S$. We recap this in Appendix~\ref{app:basics}.

\item We set $h_1(x_\ell) = 1 - e^{-\frac12 x_\ell}$. This is the missing contribution of $x_\ell$ due to the fact that it appears only in one pair.

\item $g(y_1,\ldots,y_{\ell'})$ encodes a hard instance demonstrating the impossibility of beating the factor of $1-1/e$. That is, a solution like $g(\frac{1}{\ell'},\frac{1}{\ell'},\ldots,\frac{1}{\ell'})$ should have value $1-1/e$, while the optimum  should be close to $1$.

\item $h(x_k,x_{k+1})$ is a {\em symmetry gap} instance which makes it hard to distinguish the sets associated with the variables $x_k,x_{k+1}$. More specifically, it should hold that for any solution that roughly satisfies $x_k = 2 x_{k+1}$, we have $h(x_k,x_{k+1}) = 1 - e^{-(x_k+x_{k+1})/2}$. I.e., the contributing factor depends only on $x_k+x_{k+1}$ which means that it does not distinguish between the elements of $X_k$ and $X_{k+1}$. In particular, this is true for a random solution and the types of solutions that the algorithm will be able to find with any significant probability.

\item For solutions deviating from $x_k = 2 x_{k+1}$, the value \emph{decreases}. Hence, the solutions using only the variables $x_1,\ldots,x_r$ that achieve a value of $1-1/e$ would need to satisfy $x_k = 2 x_{k+1}$ (roughly). However, this would force $x_1$ to be a non-trivially large value, which makes $f$ suffer from the $1-h(0,x_1)$ factor. This is essentially the reason why good solutions cannot be found in a limited number of rounds, although the precise analysis is a bit more complicated.
\end{itemize}

Next, we specify the function $h(x,x')$ in more detail. In the following, we work with a general parameter $\eps>0$, which will be eventually chosen to be $\eps = 2 n^{-1/24}$.

\begin{lemma}
\label{lemma:h-def}
For $\eps>0$, define $h:[0,\infty) \times [0,\infty) \rightarrow [0,1]$ as
\begin{itemize}
\item $h(x,x') = 1 - e^{-\frac12 (x+x')} $, if $|x - 2x'| \leq \eps$.
\item $h(x,x') = 1 - (\frac23 e^{-\frac34 x + \frac14 \eps} + \frac13 e^{-\frac32 x' - \frac12 \eps})$, if $x - 2x' \geq \eps$,
\item $h(x,x') = 1 - (\frac23 e^{-\frac34 x - \frac14 \eps} + \frac13 e^{-\frac32 x' + \frac 12 \eps})$, if $x - 2x' \leq -\eps$.
\end{itemize}

Then $h(x,x')$ is well-defined (the values coincide for $|x-2x'| = \eps$), its first partial derivatives are continuous, positive and decreasing in both variables.
\end{lemma}

(Note: the asymmetry between $x$ and $x'$ in the expressions comes from the fact that for a random solution we expect $x = 2x'$, which defines the ``symmetric region" here.)

\begin{proof}
For $x-2x' = \eps$, the first definition gives $h(x,x') = 1 - e^{-\frac12 (x+x')} = 1 - e^{-\frac32 x'- \frac12 \eps}$, while the second definition gives $h(x,x') = 1 - (\frac23 e^{-\frac34 x + \frac14 \eps} + \frac13 e^{-\frac32 x'- \frac12 \eps}) = 1 - e^{-\frac32 x'- \frac12 \eps}$, so the two definitions are consistent. 
Let us verify the partial derivatives now. The first definition gives
$ \partdiff{h}{x} = \frac12 e^{-\frac12(x+x')} $
while the second definition gives
$ \partdiff{h}{x} = \frac12 e^{-\frac34 x + \frac14 \eps}.$
Just like above, it is easy to verify that the two expressions coincide for $x-2x' = \eps$, and hence $\partdiff{h}{x}$ is continuous there.
Similarly, we can verify that for $x-2x' = -\eps$ the two definitions give continuous partial derivatives.
Finally, the partial derivatives are obviously positive and decreasing in each of the two variables.
\end{proof}

Next, we state the conditions that we require for the function $g(y_1,\ldots,y_{\ell'})$. We prove the following lemma in Appendix~\ref{sec:1-1/e}.

\begin{lemma}
\label{lemma:g-properties}
For any $0<\eps<1$ and $\ell' \geq 2/\eps^2$, there is a function $g:[0,1]^{\ell'} \rightarrow [0,1]$ such that
\begin{itemize}
    \item $g$ is continuous, non-decreasing in each coordinate, and it has  first partial derivatives $\partdiff{g}{y_i}$ almost everywhere\footnote{To be more precise, the partial derivatives are defined almost everywhere on any axis-aligned line, in the sense of Lemma~\ref{lem:discretize}.}
    which are non-increasing in every coordinate.
    \item We have $g(1,0,0,\ldots,0) = 1-\eps$.
    \item If $|y_i- \frac{1}{\ell'} \sum_{i=1}^{\ell'} y_i| \leq \eps/2$ for all $i \in [\ell']$,
     then $g(\by) = \min \{ 1 - e^{-\sum_{j=1}^{\ell'} y_j},  1-\eps \}.$
\end{itemize}
\end{lemma}

Recall that we use $g(y_1,\ldots,y_{\ell'})$ and $h(x,x')$ to define the function $f_1$ and $f = \min \{f_1, 1-\eps \}$ as above. The conditions stated in Lemma~\ref{lemma:h-def} and Lemma~\ref{lemma:g-properties} imply that a function obtained from $f$ by discretization is monotone submodular (see e.g. \cite{V13} for more details).

Next, we prove a quantitative bound on how much the value of $h$ (and henceforth $f$) decreases when the variables deviate from the symmetric region $x_k = 2 x_{k+1}$.

\begin{lemma}
\label{lemma:gain-per-round}
Suppose $|x - 2 x'| = \delta \geq \eps$. Then for $h(x,x')$ defined as in Lemma~\ref{lemma:h-def},
$$ h(x,x') \leq 1 - e^{-\frac12 (x+x')+\frac{1}{16}(\delta-\eps)^2}.$$
\end{lemma}

\begin{proof}
For $x-2x' = \delta \geq \eps$, we have
\begin{equation*}
\begin{split}
h(x,x') &= 1 - \left(\frac23 e^{-\frac34 x + \frac14 \eps} + \frac13 e^{-\frac32 x' - \frac12 \eps}\right) \\
 &= 1 - e^{-\frac12(x+x')} \left(\frac23 e^{-\frac14 (\delta-\eps)} + \frac13 e^{\frac12 (\delta - \eps)} \right).
\end{split}
\end{equation*}
We use two elementary bounds: $e^t \leq 1 + t + t^2$ for $|t| \leq 1$, and $e^t \leq 1+2t$ for $0 \leq t \leq 1$. Hence,
\begin{eqnarray*}
 h(x,x') & \leq & 1 - \frac23 \left(1-\frac14(\delta-\eps) + \frac{1}{16} (\delta-\eps)^2 \right)e^{-\frac12(x+x')} \\  
 & & - \frac13 \left(1 + \frac12 (\delta - \eps) + \frac14 (\delta-\eps)^2 \right) e^{-\frac12(x+x')} \\
& \leq & 1 - e^{-\frac12(x+x')+\frac{1}{16} (\delta-\eps)^2}.
\end{eqnarray*}
Similarly, we get the same bound for $x-2x'=-\delta \leq -\eps$.
\end{proof}

\begin{lemma}
\label{lemma:r-round-solution}
Suppose that an algorithm uses $r-1$ rounds of adaptivity, $r < \ell$. Then with high probability, the only solutions it can find are, up to additive error $O\left(n^{-1/24}\right)$ in each coordinate, in the form $(x_1,\ldots,x_r,\frac12 x_r,\frac14 x_r,\ldots,\frac{1}{2^{\ell-r}} x_r,y_1,\ldots,y_{\ell'})$, where $y_i = \frac{1}{\ell' 2^{\ell-r}} x_r$.
\end{lemma}

\begin{proof}
We prove the following by induction: With high probability, the computation path of the algorithm and the queries it issues in the $r$-th round are determined by $X_1,\ldots,X_{r-1}$ (and do not depend on the way that $E\setminus \cup_{i=1}^{r-1} X_i$ is partitioned into the $X_{i}$ and $Y_i$). 

As a first step, we assume the algorithm is deterministic by fixing its random bits and choose the partition of $E$ into $X_i$ and $Y_i$ uniformly at random. 

To prove the inductive claim, let $\cE_r$ denote the ``atypical event" that the algorithm issues any query $Q$ in round $r$ such that the answer is {\em not} in the form $f = \min \{\tilde{f_1}, 1-\eps \}$,
\begin{equation}
    \label{eq:sym-form}
\tilde{f_1} = 1 - e^{-\frac12 x_r - \sum_{i=r+1}^{\ell} x_i - \sum_{j=1}^{\ell'} y_j} \prod_{\ell=1}^r (1-h(x_{\ell-1},x_{\ell})), 
\end{equation}
where $x_0 = 0$, $x_i = \frac{1}{k} |Q \cap X_i|$, and $y_j = \frac{1}{k} |Q \cap Y_j|$.
Assuming that $\cE_r$ does not occur, all answers to queries in round $r$ are in this form, and in particular they depend only on $Q$ and the sets $X_1,\ldots,X_r$. (The summation $\sum_{i=r+1}^{\ell} x_i + \sum_{j=1}^{\ell'} y_j$ is determined by $|Q \setminus (X_1 \cup \ldots \cup X_r)|$.) If the queries in round $r$ depend only on $X_1,\ldots,X_{r-1}$, and $\cE_r$ does not occur, this means that the entire computation path in round $r$ is determined by $X_1,\ldots,X_r$. By induction, we conclude that if none of $\cE_1,\ldots,\cE_r$ occurs, the computation path in round $r$ is determined by $X_1,\ldots,X_r$.

In the following we focus on the analysis of the event $\cE_r$.
Let $\cQ_r$ denote the queries in round $r$, assuming that none of $\cE_1,\ldots,\cE_{r-1}$ occurred so far. $\cQ_r$ is determined by $X_1,\ldots,X_{r-1}$. Conditioned on $X_1,\ldots,X_{r-1}$, the partitioning of $E \setminus \cup_{i=1}^{r-1} X_{i}$ is uniformly random.
This implies that for each query $Q$, the set $Q \setminus (X_1 \cup \ldots \cup X_{r-1})$ is partitioned randomly into $Q \cap X_r, \ldots, Q \cap X_\ell, Q \cap Y_1, \ldots, Q \cap Y_{\ell'}$ and the cardinalities $|Q \cap X_i|, |Q \cap Y_j|$ are concentrated around their expectations. We have $\E[|Q \cap X_{i+1}|] = \frac12 \E[ |Q \cap X_{i}|]$ and $\E[ |Q \cap Y_j| ] = \frac{1}{\ell'} \E[|Q \cap Y|] = \frac{1}{\ell'} \E[|Q \cap X_\ell|]$ for any $r \leq i < \ell$ and $1 \leq j \leq \ell'$.
By Hoeffding's bound\footnote{Technically, Hoeffding's bound does not apply directly, since elements do not appear in $X_1,X_2,\ldots$ independently. However, due to the cardinality constraints, the appearances of elements are negatively correlated, so Hoeffding's bound still applies; see \cite{PS97}.},
for $x_i = \frac{1}{k} |Q \cap X_i|$, $i \geq r$,  and conditioned on the choice of $X_1,\ldots,X_{r-1}$,
\begin{equation*}
\begin{split}
    \prob{\left| x_i - \E x_i \right| > \alpha} & \leq 2 \exp\left(-\alpha^2 k^2 / |X_i| \right)  \leq 2 \exp\left(-\alpha^2 n^{1/4}\right),
\end{split}
\end{equation*}
where we used $k = n^{2/3} / \ell' = n^{5/8}$ and $|X_i| \leq n/2$.
Similarly, $|y_j - \E y_j| > \alpha$ with probability at most $2 e^{-\alpha^2 n^{1/4}}$.
We set $\alpha = \frac12 \eps = n^{-1/24}$ to obtain a high probability bound in the form $1 - e^{-\Omega(n^{1/6})}$.

If $|x_i - \E x_i| \leq \eps/2$ and $|y_j - \E y_j| \leq \eps/2$ for all $i \geq r$ and $j \leq \ell$, this means that the query $Q$ is in the symmetric region for all the relevant evaluations of $h(x_i,x_{i+1})$ and $g(y_1,\ldots,y_{\ell'})$. Therefore, by construction the answer will be in the form of Equation~\ref{eq:sym-form}, which only depends on $Q$ and $X_1,\ldots,X_r$. 

Let us bound the probability of $\cE_r \setminus ({\cE_1} \cup {\cE_2} \cup \ldots \cup {\cE_{r-1}})$. If we condition on $X_1,\ldots,X_{r-1}$, assuming that none of $\cE_1,\ldots,\cE_{r-1}$ occurred, the query set $\cQ_r$ in round $r$ is fixed. By a union bound over $\cQ_r$, the probability that any of them violates Equation~\ref{eq:sym-form} is $e^{-\Omega(-n^{1/6})}$.
Hence, 
$$ \Pr[\cE_r \setminus ({\cE_1} \cup \ldots \cup {\cE_{r-1}}) \mid X_1,\ldots,X_{r-1}] \leq \text{poly}(n) \ e^{-\Omega(n^{-1/6})}.$$
Now we can average over the choices of $X_1,\ldots,X_{r-1}$ and still obtain
$ \Pr[\cE_r \setminus \cup_{i=1}^{r-1}{\cE_i}] = e^{-\Omega(n^{-1/6})}.$
Therefore, by induction,
\begin{equation*}
\begin{split}
\Pr\left[\bigcup_{i=1}^{r} \cE_i \right] &= \Pr[\cE_1] + \Pr[\cE_2 \setminus \cE_1] + \ldots + \Pr[\cE_r \setminus \cup_{i=1}^{r-1}{\cE_i}] \\
& = r e^{-\Omega(n^{-1/6})} = e^{-\Omega(n^{-1/6})}.
\end{split}
\end{equation*}
This implies that with high probability, the computation path in round $r$ is determined by $X_1,\ldots,X_{r-1}$.

Consequently, a solution returned after $r-1$ rounds is determined by $X_1,\ldots,X_{r-1}$ with high probability.
By the same Chernoff-Hoeffding bounds, the solution is with high probability in the form
$$\left(x_1, \ldots, x_r, \frac12 x_r,\frac14 x_r,\ldots,\frac{1}{2^{\ell}} x_r,y_1,\ldots,y_{\ell'} \right),\ \  \, y_i = \frac{1}{\ell' 2^{\ell-r}} x_r,$$
up to additive error $\pm n^{-1/24}$ in each coordinate.

Finally, we note that by allowing the algorithm to use random bits, the results are a convex combination of the bounds above, so the same high probability bounds are satisfied.
\end{proof}

We remark an adaptive algorithm can indeed learn the identity of $X_1,\ldots,X_r$ in the first $r$ rounds, so in this sense our analysis is tight. In the $r$-th round, we can determine $X_r$ by examining the marginal values of elements with respect to a random subset of $E \setminus (X_1 \cup \ldots \cup X_{r-1})$.  After the $r$ rounds, an adaptive algorithm can completely determine $\{X_i\}_{i \leq r}$ and is free to choose the values of $x_1,\ldots,x_{r}$ in a query, but not the further variables.

\subsection{Analysis of an \texorpdfstring{$r$}{r}-round algorithm}

Here we bound the value that an algorithm can possibly achieve in $r$ rounds.

\begin{lemma}
\label{lemma:quadr-opt}
The optimum of the following optimization problem has value at least $\frac{1}{4r}$:
\begin{equation*}
\begin{split}
    \min_{\mathbf{x}} & \,  4x_1^2 + \sum_{i=2}^{r} \left( 2 x_{i} - x_{i-1} \right)^2: 
     \ \ \ \sum_{i=1}^{r-1} x_i + 2x_r \geq \frac12, \,  x_i \geq 0.
\end{split}
\end{equation*}
\end{lemma}

\begin{proof}
Denote the objective by $\varphi(\mathbf{x}) := \sum_{i=1}^{r} \left( 2 x_{i} - x_{i-1} \right)^2$ where $x_0 = 0$.
By Cauchy-Schwarz, we have
\begin{equation*}
\varphi(\mathbf{x}) \geq  \frac1r \left( \sum_{i=1}^{r} (2x_i-x_{i-1}) \right)^2
 = \frac1r \left( \sum_{i=1}^{r-1} x_i + 2 x_r \right)^2
\geq \frac{1}{4r}. \qedhere
\end{equation*}
\end{proof}

\begin{theorem}
\label{thm:main-log-monotone}
Any $r$-round adaptive algorithm for monotone submodular optimization can achieve at most a $1 - 1/e - \Omega(1/r)$ approximation, for $r \leq \frac13 \log_2 n$ where $n$ is the number of elements.
\end{theorem}

\begin{proof}
By Lemma~\ref{lemma:r-round-solution}, in $r$ rounds we can only find solutions of the form $$(x_1, \ldots, x_r, \frac12 x_r,\frac14 x_r,\ldots,\frac{1}{2^{\ell}} x_r,y_1,\ldots,y_{\ell'}),\ \ \, y_i = \frac{1}{\ell' 2^{\ell-r}} x_r$$ up to $\pm n^{-1/24}$ error in each coordinate. Choosing $\eps = 2n^{-1/24}$, any solution found after $r$ rounds is w.h.p.~such that the objective function $f = \min \{f_1, 1-\eps \}$ has the same value as the function $\tilde{f} = \min \{\tilde{f_1}, 1-\eps\}$, where
$$\tilde{h} = (1 - h(0, x_1))(1 - h(x_1, x_2))\cdots(1 - h(x_{r-2}, x_{r-1})),$$
$$\tilde{s} = \exp\left(- \sum_{i=r}^\ell x_i - \sum_{i=1}^{\ell'} y_i\right),$$
$$\tilde{f_1}= 1 - \tilde{h} \tilde{s} \exp(-x_{r-1}/2).$$
(Note that the contributions of $h(x_{r-1},x_r), h(x_r,x_{r+1}),$ etc. have been replaced by their ``symmetrized variants" here.)
By Lemma~\ref{lemma:gain-per-round}, $h(x_{i-1},x_{i}) \leq 1 - e^{-\frac12 (x_{i-1}+x_{i})+\frac{1}{16}(|2x_i-x_{i-1}|-\eps)^2}$ and so
\begin{equation}
\label{eq:overall-loss}
\begin{split}
\tilde{f_1} & \leq 1 - \tilde{s}\exp\left(\frac{1}{16} (2x_1-\eps)^2+ \frac{1}{16} \sum_{i=2}^{r} (|2x_{i} - x_{i}| - \eps)^2\right) \\
& \leq 1 - \tilde{s}\exp\left(\frac{1}{64r}-\eps \right) \\
& \leq 1 - 1/e - \Omega (1/r).
\end{split}
\end{equation}
The inequality is derived as follows: Since $\tilde{f_1}$ is monotone, we can assume that the solution has maximum possible cardinality, which means $\sum_{i=1}^{\ell} x_i + \sum_{j=1}^{\ell'} y_j = 1$. We then use 
the cardinality constraint to bound $\sum_{i=1}^{r-1} x_i + 2x_r \geq \sum_{i=1}^{\ell} x_i + \sum_{j=1}^{\ell'} y_j - O(\ell' \eps) \geq 1/2$ w.h.p.,
Lemma~\ref{lemma:quadr-opt} to estimate $4x_1^2 + \sum_{i=2}^{r} (2x_i - x_{i-1})^2 \geq \frac{1}{4r}$,
and the fact that $\eps = O(n^{-1/24})$ which is negligible compared to $1/r$.
The optimal solution is any $Y_i$, which gives $f(0,0,0,\ldots,0,1) = 1 - \eps = 1 - O(n^{-1/24})$.  Thus in $r$ rounds we get an approximation factor of $1 - 1/e - \Omega(1/r)$.
\end{proof}

\section{A Poly-round lower bound}
\label{sec:poly-lower-bound}

In this section we show a variation on the lower bound in Section~\ref{sec:log-lower-bound}, choosing blocks decreasing by factors of $1+\Theta\left(\frac{\log n}{r}\right)$ instead of $2$. This will allow us to extend our result to $r=O(n^c)$ for some constant $c$ but with a weaker dependence of the loss in approximation ratio $\Theta\left(\frac{\log^2n}{r^3}\right)$ rather than $\Theta(1/r)$. 

\begin{theorem}
\label{thm:monontone-lower-bound}
For any $n$ and $r$ satisfying $\Omega(\log n) < r < O(n^{c})$, where $c > 0$ is some absolute constant, there is no algorithm using $r$ rounds of queries and achieving better than a $\left(1-1/e-\Omega\left(\frac{\log^2n}{r^3}\right)\right)$-approximation for monotone submodular maximization subject to a cardinality constraint (on a ground set of size $n$).
\end{theorem}

Instead of modifying the presentation of the previous section, we show an alternative construction that -- on the surface -- looks quite different. Our main reason for showing this alternate construction is that it is technically simpler, although less intuitive to derive.

\subsection{Construction of the objective function}

One difference in this construction is that a separate instance is needed each value of $r$, since the shrinkage rate between the blocks depend on $r$.

Let $E = X_1 \cup \ldots \cup X_r \cup Y_1 \cup \ldots Y_{\ell'}$ where $|X_{i+1}| = |X_i| / (1+\delta)$ and $|Y_i| = \frac{1}{\ell'}|X_r|=k$.\footnote{We ignore the issue of rounding number to the nearest integer. It is easy to verify that this does not affect the analysis significantly.} In the construction, we require $k = \Omega(n^{2/3})$, so we choose $\ell'= \Theta(n^{1/5})$ and $\delta = \frac{2}{15} \frac{\log n}{r}$.

For the lower bound, we consider functions in the following form:
$$ f_1=1-(1-q(x_1,\ldots,x_r)) (1-g(y_1,\ldots,y_{\ell'}))$$
where $x_i = \frac{1}{k} |S \cap X_i|$ and $y_i = \frac{1}{k} |S \cap Y_i|$.

By Lemmas proven in Appendix~\ref{app:basics}, $f$ is monotone submodular so long as $q$ and $g$ are monotone submodular. As in the previous section, $g$ will be the $(1-1/e)$-hard instance constructed in Appendix~\ref{sec:1-1/e}. Furthermore, the actual objective function will be $f = \min\{f_1, 1-\epsilon \}$ for a parameter $\eps$ to be specified later. Now we specify the function $q(x_1,\ldots,x_r)$ (we set $x_0 = 0$).
$$q(x_1,\ldots,x_r)= 1-\exp\left(-\sum_{i=1}^{r}x_i+\sum_{i=0}^{r-1}h((1+\delta)x_{i+1}-x_{i})\right),$$
\begin{displaymath}
h(x) = \left\{
\begin{array}{ll}
0 & x\le \epsilon \\
\alpha(x-\epsilon)^2 & \epsilon < x \leq 2+\epsilon \\
4\alpha(x-1-\epsilon) & x > 2+\epsilon \\
\end{array}\right.
\end{displaymath}
where $\alpha$ is a small enough constant and $\epsilon = n^{-1/10}$.

Though somewhat unintuitive, this choice of $q$ yields a hard instance $f$ with proper choice of $\alpha$ and $\eps$. We will show that this instance is both monotone and submodular. Before we explain how to choose the constants, we first explain the connection between the two constructions.

\subsection{Connection between the two constructions} 
Ignoring for now the third case of $h$, we see that $q$ can be more succintly phrased as 
\begin{equation*}
q(x_1,\ldots, x_r) = 1 - \exp\left(-\sum_{i=1}^{r-1}x_i + \alpha\sum_{i=0}^{r-1} ((1+\delta)x_{i+1} - x_i - \eps)_+^2 \right),
\end{equation*}
where $(x)_+ = \max(x, 0)$ and $x_0 = 0$.

Supposing for a moment that we set $\alpha = \frac{1}{16}$ and $\delta = 1$, this closely mimics the bound on the penalty function from Lemma~\ref{lemma:gain-per-round} and Equation~\ref{eq:overall-loss}. The main difference is that only the case $(1+\delta)x_{i+1} - x_i > \eps$ is penalized.

For a query $Q$ with $|Q|=k$ with no knowledge of any of the partitions, we expect for $(1+\delta)x_{i+1} - x_i$ to be less than $\eps$, hiding all but the first term of the penalty function. As we learn more parts, the algorithm can spread the penalty terms among the learned layers $X_i$, thus lessening the penalty (as in Section~\ref{sec:log-lower-bound}).

\subsection{Overview of the lower bound}
Unfortunately, some technicalities remain in the construction, which requires case 3 in the definition of $h$ as well as a judicious choice of $\alpha$. The following properties are true for $h$ and $q$, for some constant $\alpha>0$.
\begin{lemma}{Properties of the function $h$.}
\label{lem:h-properties}
~
\begin{itemize}
    \item $h$ is continuous, non-decreasing and differentiable.
    \item The derivative of $h$ is continuous and at most $4\alpha$.
    \item $h(x)\le 4 \alpha x$ when $x\ge0$.
\end{itemize}
\end{lemma}

\begin{proof}
The first property is easy to verify. And since the derivative is non-decreasing, it is bounded by the third case, which is $4\alpha$. For the third property, note that $4\alpha x$ coincides with $h(x)$ at $x=0$, but has derivative greater or equal to $h(x)$ for all $x$.
\end{proof}

\begin{figure}[ht]
\begin{center}
    \includegraphics[width=.95\columnwidth]{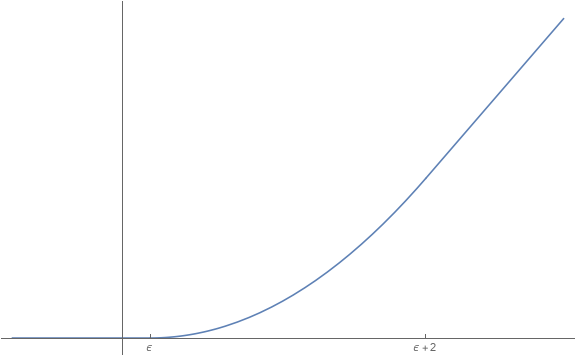}
\end{center}
\caption{Plot of the function $h$.}
\end{figure}
In particular, case 2 comes into play when proving bounds on the derivatives of $q$. This boundedness is required to show that $q$ has non-increasing first partial derivatives.

\begin{lemma}
\label{lem:q-properties}
For $\alpha \in (0,\frac{1}{16})$, $\delta \in (0,1)$,
the function $q$ satisfies
\begin{itemize}
\item $q(x_1, \ldots, x_{r}, y_1, \ldots, y_{\ell'}) \in [0,1]$.
\item $q$ is continuous and it has first partial derivatives which are continuous in every coordinate.
\item $q$ is non-decreasing and its first partial derivative are non-increasing in each coordinate.
\end{itemize}
\end{lemma}

\begin{proof}
Let us denote $q(x_1,\ldots,x_r) = 1 - e^{-p(x_1,\ldots,x_r)}$, $$ p(x_1,\ldots,x_r) = \sum_{i=1}^{r}x_i-h((1+\delta)x_1)-\sum_{i=1}^{r-1}h((1+\delta)x_{i+1}-x_{i}).$$
First we prove that when each coordinate is non-negative, $0 \leq q(x_1,\ldots, x_r) \leq 1$. The second inequality is obvious. In order to prove $1-\exp(-p(x_1,\ldots, x_r))$ is non-negative, we only need to prove that $p(x_1, \ldots, x_r)$ is non-negative.  By Lemma~\ref{lem:h-properties}, $h$ is non-decreasing and $h(x)\le 4\alpha x$ when $x\ge0$, so we have 

\begin{align*}
p(x_1,\ldots, x_r)&=\sum_{i=1}^{r}x_i-h((1+\delta)x_1)-\sum_{i=1}^{r-1}h((1+\delta)x_{i+1}-x_{i})\\
&\ge \sum_{i=1}^{r}x_i-h((1+\delta)x_1)-\sum_{i=1}^{r-1}h((1+\delta)x_{i+1}) \\
&\ge \sum_{i=1}^{r}x_i-4\alpha(1+\delta)\sum_{i=1}^rx_i\ge 0
\end{align*}
using $0 < \alpha \leq \frac{1}{16}$ and $0 < \delta \le 1$ in the last inequality.

Next we note that the second property of $q$ is a direct corollary of Lemma~\ref{lem:h-properties}. More specifically, by Lemma~\ref{lem:h-properties}, $h$ is continuous, differentiable and its first order derivative is continuous. Thus, by definition $q(x_1,\ldots, x_r)$ is also continuous and its first partial derivative  is continuous in each coordinate. 

Finally we prove that $q(x_1,\ldots, x_r)$ is non-decreasing and its first partial derivatives are non-increasing in every coordinate. By Lemma~\ref{lem:h-properties}, the first order derivative of $h$ is at most $4\alpha$. So by definition of $p$, we have 
$$\frac{\partial p}{\partial x_i}\ge 1-4\alpha(1+\delta)$$
For the first-order partial derivative of $q$,
$$\frac{\partial q}{\partial x_i}=\frac{\partial p}{\partial x_i} e^{-p(x_1,\ldots, x_r) }\ge(1-4\alpha(1+\delta)) e^{-p(x_1,\ldots, x_r)}\ge0$$ where we use $0 < \alpha \leq \frac{1}{16}$ and $0 < \delta\le 1$ in the last inequality. Thus $q$ is non-decreasing. By definition of $p$, the second partial derivative $\frac{\partial^2 p}{\partial x_i\partial x_j}$ is non-zero if and only if $x_i$, $x_j$ share the same $h$ and that $h$ is quadratic. Therefore $\frac{\partial^2 p}{\partial x_i\partial x_j}=2\alpha(1+\delta)$ (when $j = i+1$) or 0. So for the second partial derivative of $q$, 
$$\frac{\partial^2 q}{\partial x_i\partial x_j}=\left(\frac{\partial^2 p}{\partial x_i\partial x_j}-\frac{\partial p}{\partial x_i}\cdot\frac{\partial p}{\partial x_j}\right)e^{-p(x_1,\ldots, x_r)}$$
$$\le (2\alpha(1+\delta)-(1-4\alpha(1+\delta))^2) e^{-p(x_1,\ldots, x_r)}\le 0$$
for $0 < \alpha \leq \frac{1}{16} \leq \frac{1}{8(1+\delta)}$ (for $0< \delta\le 1$). 
Thus the first partial derivatives of $q$ are non-increasing.
\end{proof}

These two lemmas along with Lemma~\ref{lem:monotone-combine} imply that $f$ is monotone submodular. Next, we show that in $r$ rounds, the best approximation we can achieve is $1-1/e-\Omega\left(\frac{\log^2 n}{r^3}\right)$.
 
\subsection{Analysis of an \texorpdfstring{$r$}{r}-round algorithm}

The following analogue of Lemmas~\ref{lemma:r-round-solution},~\ref{lemma:quadr-opt}, and Theorem~\ref{thm:main-log-monotone} can be shown.
\begin{lemma}
\label{lem:r-round-solution}
Suppose that an algorithm uses $s-1$ rounds of adaptivity, $s < r$. Then with high probability, the only solution it can find  $(x_1, \ldots, x_r, y_1, \ldots, y_{\ell'})$ satisfies the properties that $(1+\delta)x_{i+1}-x_i\le\epsilon$ for any $s\le i\le r-1$ , $|y_i-\bar{y}|\le\epsilon/2$ for any $1\le i\le \ell'$ and $|x_r-\ell'\bar{y}|>\epsilon$. 
\end{lemma}

The proof of this lemma will mimic the proof of Lemma~\ref{lemma:r-round-solution}.

\begin{proof}
We first remind the reader that the parameters of our construction are $\delta = \frac{2}{15} \frac{\log n}{r}$ and $\ell' = \Theta(n^{1/5})$. This means that $k = \frac{n}{\ell' (1+\delta)^r} = \Omega(n^{2/3})$.

Continuing with the proof, we can assume by Yao's principle that the algorithm is deterministic.

We prove the following by induction: With high probability, the computation path of the algorithm and the queries it issues in the $s$-th round are determined by $X_1,X_2,\ldots,X_{s-1}$ (and do not depend on the way that $X_{s} \cup X_{s+1} \cup \ldots \cup X_r \cup Y_1 \cup \ldots \cup Y_{\ell'}$ is partitioned into $X_{s}, \ldots, X_r, Y_1,\ldots, Y_{\ell'}$). 


To prove the inductive claim, let $\cE_s$ denote the ``atypical event" that the algorithm issues any query $Q$ in round $s$ such that the answer is {\em not} in the form $\tilde{f} = \min \{\tilde{f}_1, 1-\epsilon\}$,

\begin{equation} 
\label{eq-symmetry}
\tilde{f}_1 = 1-\exp\left(-\sum_{i=1}^{r}x_i-\sum_{i=1}^{\ell'} y_i+\sum_{i=0}^{s-1}h((1+\delta)x_{i+1}-x_i)\right),
\end{equation}
where $x_0 = 0$, $x_i = \frac{1}{k} |Q \cap X_i|$, and $y_j = \frac{1}{k} |Q \cap Y_j|$.
Assuming that $\cE_s$ does not occur, all answers to queries in round $s$ are in this form, and in particular they depend only on $Q$ and the sets $X_1,\ldots,X_s$. (The summation $\sum_{i=s+1}^{r} x_i + \sum_{j=1}^{\ell'} y_j$ is determined by $|Q \setminus (X_1 \cup \ldots \cup X_s)|$.) Assuming that the queries in round $s$ depend only on $X_1,\ldots,X_{s-1}$, and $\cE_s$ does not occur, this means that the entire computation path in round $s$ is determined by $X_1,\ldots,X_s$. By induction, we conclude that if none of $\cE_1,\ldots,\cE_s$ occurs, the computation path in round $s$ is determined by $X_1,\ldots,X_s$.

In the following we focus on the analysis of the event $\cE_s$.
Let $\cQ_s$ denote the queries in round $s$, assuming that none of $\cE_1,\ldots,\cE_{s-1}$ occurred so far. $\cQ_s$ is determined by $X_1,\ldots,X_{s-1}$. Conditioned on $X_1,\ldots,X_{s-1}$, the partitioning of $E \setminus \cup_{i=1}^{s-1} X_i$ is uniformly random.
This implies that for each query $Q$, the set $Q \setminus (X_1 \cup \ldots \cup X_{s-1})$ is partitioned randomly into $Q \cap X_s, \ldots, Q \cap X_r, Q \cap Y_1, \ldots, Q \cap Y_{\ell'}$ and the cardinalities $|Q \cap X_i|, |Q \cap Y_j|$ are concentrated around their expectations. We have $(1+\delta)\E[|Q \cap X_{i+1}|] =  \E[ |Q \cap X_{i}|]$ and $\E[ |Q \cap Y_j| ] = \frac{1}{\ell'} \E[|Q \cap Y|] = \frac{1}{\ell'} \E[|Q \cap X_r|]$ for any $s \leq i < r$ and $1 \leq j \leq \ell'$.
By Hoeffding's bound,
for $x_i = \frac{1}{k} |Q \cap X_i|$, $i \geq s$,  and conditioned on the choice of $X_1,\ldots,X_{s-1}$,
\begin{equation}
\begin{split}
    \prob{(1+\delta) x_{i+1} - x_i  > \epsilon} & \leq \exp\left(-\frac{2\epsilon^2 k^2}{|X_i|+(1+\delta)^2 |X_{i+1}|}\right) \\
    & \leq  \exp\left(- n^{1/10}\right),
\end{split}
\end{equation}
where we use $k = \Omega(n^{2/3})$, $|X_i|+(1+\delta)^2|X_{i+1}| \leq 2n$ and $\epsilon=n^{-1/10}$.
Similarly, we can prove the same bound for $|y_j - \bar{y}| > \epsilon/2$ and $|x_r-\ell'\bar{y}|>\epsilon$,
$$\Pr[|y_j-\bar{y}|>\epsilon/2]\le 2\exp\left(-\frac{-\epsilon^2k^2}{2|Y|}\right)\le \exp(-n^{1/10}),$$
$$\Pr[|x_r-\ell'\bar{y}|>\epsilon]\le 2\exp\left(-\frac{-2\epsilon^2k^2}{|Y|+|X_r|}\right)\le \exp(-n^{1/10}),$$

For $s\le i\le r-1$, $h((1+\delta)x_{i+1}-x_i)=0$ because $(1+\delta)x_{i+1}-x_i\le\epsilon$, $g(y_1, \ldots, y_{\ell'})=\min\{1-\epsilon, 1-\exp\left(\sum_{i=1}^{\ell'}y_i\right)\}$ because $|y_i-\bar{y}|\le\epsilon/2$ for any $1\le i\le \ell'$. If $g(y_1, \ldots, y_{\ell'}) = 1-\epsilon$, then $f_1(x_1, \ldots, x_r, y_1, \ldots, y_\ell')\ge1-\epsilon$. Therefore, by construction the answer will be the form Equation~\ref{eq-symmetry}, which only depends on $Q$ and $X_1,\ldots,X_s$. 

Let us bound the probability of $\cE_r \cap \overline{\cE_1} \cap \overline{\cE_2} \cap \ldots \cap \overline{\cE_{s-1}}$. If we condition on $X_1,\ldots,X_{s-1}$, assuming that none of $\cE_1,\ldots,\cE_{s-1}$ occurred, the query set $\cQ_s$ in round $s$ is fixed. By a union bound over $\cQ_s$, the probability that any of them violates Equation~\ref{eq-symmetry} is $e^{-\Omega(n^{1/10})}$.
Hence, 
$$ \Pr[\cE_s \setminus \cup_{i=1}^{s-1} {\cE_i} \mid X_1,\ldots,X_{s-1}] = \text{poly}(n) \ e^{-\Omega(n^{1/10})}.$$
Now we can average over the choices of $X_1,\ldots,X_{s-1}$ and still obtain
$$ \Pr[\cE_s \setminus ({\cE_1} \cup {\cE_2} \cup \ldots \cup {\cE_{s-1}})] = e^{-\Omega(n^{1/10})}.$$
Therefore, by induction,
\begin{equation*}
\begin{split}
\Pr\left[\bigcup_{i=1}^{s} \cE_i \right] & = \Pr[\cE_1] + \Pr[\cE_2 \setminus \cE_1] + \ldots + \Pr[\cE_s \setminus \cup_{i=1}^{s-1} \cE_i] \\
& = s e^{-\Omega(n^{1/10})} = e^{-\Omega(n^{1/10})}.
\end{split}
\end{equation*}
This implies that with high probability, the computation path in round $s$ is determined by $X_1,\ldots,X_{s-1}$.

Consequently, a solution returned after $s-1$ rounds is determined by $X_1,\ldots,X_{s-1}$ with high probability.
By the same Chernoff-Hoeffding bounds, the solution with high probability satisfies the properties that $(1+\delta)x_{i+1}-x_i\le\epsilon$ for any $s\le i\le r-1$ , $|y_i-\bar{y}|\le\epsilon/2$ for any $1\le i\le \ell'$ and $|x_r-\ell'\bar{y}|>\epsilon$.
\end{proof}

Next we bound the value that an algorithm can possibly achieve in $r-1$ rounds. 

\begin{lemma}
\label{lem:polyround-opt}
When $\omega(\epsilon) = \delta/r = o(1)$, the optimum of the following optimization problem has value $\Omega(\delta^2 / r)$:
\begin{equation}
\begin{split}
    \min_{\mathbf{x}} & \, \sum_{i=0}^{r-1}h((1+\delta)x_{i+1}-x_{i})  : \\
     & \, \sum_{i=1}^{r-1} x_i + x_r \ge \frac13,  \, x_i \ge 0, x_0 = 0 .
\end{split}
\end{equation}
\end{lemma}

\begin{proof}

By the convexity and monotonicity of $h$, 

\begin{align*}
& \frac{1}{r} \sum_{i=0}^{r-1}h((1+\delta)x_{i+1}-x_{i}) \\
\, & \ge h\left(\frac1r(1+\delta)x_1+\frac1r\sum_{i=2}^{r}((1+\delta)x_i-x_{i-1})\right) \\
\, & = h\left(\frac1r\delta\sum_{i=1}^r x_i + \frac1r(1-\delta) x_r\right) \\
\, & \ge h\left(\frac{\delta}{3r}\right).
\end{align*}

Recall that when $x\in [\epsilon, 2+\epsilon]$, $h(x) = \alpha(x-\epsilon)^2$ where $\alpha>0$ is a constant. By the assumption in this lemma, $\frac{\delta}{3r} \in [\omega(\epsilon), o(1)]$, so the optimum of this program is at least $$rh\left(\frac{\delta}{3r}\right) = r\alpha \left(\frac{\delta}{3r}-\epsilon \right)^2 = \Omega\left(\frac{\delta^2}{r}\right).$$

\end{proof}

Finally, we prove the main theorem of this section, 

\begin{proof}[Proof of Theorem~\ref{thm:monontone-lower-bound}]
We construct the function $f = \min \{1-\epsilon, f_1\}$ as before. $f$ is non-decreasing, and without loss of generality we can assume $\sum_{i=1}^rx_i+\sum_{i=1}^{\ell'} y_i=1$ for the solution. By Lemma~\ref{lem:r-round-solution},  we know that, with high probability, the answer output by an $(r-1)$-round algorithm will be in the form $(x_1, \ldots, x_r, y_1, \ldots, y_{\ell'})$ such that $|x_r-\ell'\bar{y}|<\epsilon$.  Thus we have 
$$\sum_{i=1}^r x_i\ge \frac12\left(\sum_{i=1}^r x_i+\sum_{i=1}^{\ell'}    y_i-\ell'\epsilon\right)=\frac12(1-\epsilon)\ge1/3.$$

And recall that we define our parameters in the following way, $\epsilon = n^{-1/10}$, $\Omega(\log n) = r = O(n^c)$ and $\delta = \frac{\log n}{r}$. Thus it is easy to verify that $\omega(\epsilon) = \delta/r = o(1)$ if we choose $c>0$ as a small enough constant.  As a result,

\begin{align*}
p(x_1, \ldots, x_r, y_1, \ldots, y_{\ell'}) &
= \sum_{i=1}^{r}x_i-\sum_{i=0}^{r-1}h((1+\delta)x_{i+1}-x_{i}) \\
& \le \sum_{i=1}^{r}x_i - \Omega\left(\frac{\delta^2}{r}\right)
\end{align*}
where $x_0 = 0$ and we use Lemma~\ref{lem:polyround-opt} in the inequality. We can further bound $f(x_1, \ldots, x_r, y_1, \ldots, y_{\ell'})$,
\begin{align*}
f&=1-\exp(-p(x_1,\ldots,x_r)) (1-g(y_1,\ldots,y_{\ell'}))\\
&=1-\exp(-p(x_1,\ldots,x_r)) \exp\left(-\sum_{i=1}^{\ell'}y_i\right) \\
&\le 1-\exp\left(-\sum_{i=1}^{r}x_i-\sum_{i=1}^{\ell'}y_i +\Omega\left(\frac{\delta^2}{r}\right)\right)\\
&= 1-1/e-\Omega\left(\frac{\delta^2}{r}\right)
\end{align*}
where $\delta = O\left(\frac{\log n}{r}\right)$. By Lemma~\ref{lemma:1-1/e}, we know that $OPT\ge 1-\epsilon$. Since $\epsilon = n^{-1/10} = o\left(\frac{\log^2n}{r^3}\right)$,  no algorithm can achieve approximation ratio better than $1-1/e-\Omega\left(\frac{\log^2n}{r^3}\right)$ in $r-1$ rounds with high probability.
\end{proof}
\section{Improved analysis for unconstrained non-monotone maximization}
\label{sec:non-monotone-alg}
In this section we show the following result:

\begin{theorem}
\label{thm:dg-gains}
Let $R\subseteq E$ be a uniformly random subset and $f:2^E \rightarrow \RR_+$ be a non-monotone submodular function with maximum value $\textrm{OPT}$. If $\E[f(R)] \leq (1/2 - \delta)OPT$, then the low-adaptivity continuous double greedy algorithm (Algorithm~\ref{alg:cont-dg}) achieves value at least $(1/2 + \Omega(\delta^2))OPT$. Furthermore, the algorithm achieves this value in  $O(1/\delta^2)$ rounds.
\end{theorem}

As previously mentioned, the current state-of-the-art algorithm for unconstrained non-monotone maximization takes ${O}(1/\eps)$ rounds to get a $(1/2-\eps)$-approximation~\cite{ENV19, CFK19}. While we do not improve this result, we show that for an instance with given $\E[f(R)] = (1/2-\delta) OPT$,
the number of rounds doesn't blow up arbitrarily as we approach the factor of $1/2$ (in contrast to the monotone cardinality-constrained problem). It takes $O(\min\{\frac{1}{\delta^2}, \frac{1}{\epsilon} \})$ rounds to achieve a $(1/2-\epsilon)$-approximation, and in fact a strictly {\em better} than $1/2$ approximation in $O(1/\delta^2)$ rounds.

The main intuition for this result is that in some sense, the worst possible sequence of steps in double greedy returns exactly the $\frac{1}{2}\mathds{1}$ point. In this case, the analysis of existing algorithms show no gain over $OPT/2$~\cite{ENV19}. However, this is also exactly the value of a random set, which we can evaluate in just one adaptive round.

As a counterpart to our positive result, we show that there are instances where the random solution has value $(1/2 - \Theta(\delta)) OPT$ and it is difficult for any polynomial round algorithm to get better than $(1/2 + \delta) OPT$.

\begin{lemma}
Let $R\subseteq E$ be a uniformly random subset For any $\delta > 0$, there exist submodular functions $f$ such that $\E[f(R)] = (1/2-\Theta(\delta)) OPT$ and no algorithm can achieve better than a $1/2 + O(\delta)$-approximation in a polynomial number of rounds. 
\end{lemma}

\begin{proof}
One can construct such instances quite easily using the hardness instances of Vondrak et al.~\cite{FMV11}. Let $f_{1/2}$ be a $1/2$-hardness instance for non-monotone submodular optimization in the value query model, as defined in Section 4.2 of ~\cite{FMV11}. The main properties of $f_{1/2}$ is that (1) a uniformly random set $R$ has value $\E[f(R)] = \left(\frac{1}{2}-o(1)\right) OPT$, and (2) no algorithm using polynomially many queries to $f$ can find a solution better than $(1/2 + o(1)) OPT$. 

To construct our hardness instance, we combine this hard instance with a directed cut instance where a random set has value $\frac14 OPT$.
We first partition the ground set $E$ randomly into two halves $X_1$ and $X_2$, and let $f_{\delta}(S) = \delta x_1 (1-x_2) OPT$, where $x_i = |S\cap X_i| / |X_i|$. Our hardness instance is then simply 
$$f(S) = f_{1/2}(S) + f_{\delta}(S).$$
For a uniformly random subset $R$, we have $\E[f(R)] = (1/2 + \delta/4) OPT$. On the other hand, the optimal solution to $f$ has value $(1 + \delta) OPT$. However, the properties of $f_{1/2}$ guarantee that no algorithm using polynomially many queries can obtain a value better than $(1/2 + \delta + o(1)) OPT$. Thus, relative to the optimum of $f$, the random set obtains a $\frac{1/2 + \delta/4}{1+\delta} = (1/2 - \Theta(\delta))$-approximation, and no polynomial-round adaptive algorithm can achieve better than a $(1/2 + O(\delta))$-approximation.
\end{proof}

\subsection{Continuous double greedy}
As a preliminary, we review the low-adaptivity continuous double greedy (Algorithm~\ref{alg:cont-dg}) of Ene et al.~\cite{ENV19} with some modifications. The algorithm assumes access to the multilinear extension of $f$ (this assumption can be removed through standard sampling techniques). We use the notation $\mathbf{v}_+$ and $\mathbf{v}_-$ to denote coordinate-wise $\max(\mathbf{v}, \mathbf{0})$ and $\min(\mathbf{v}, \mathbf{0})$ respectively. Other arithmetic operations in Algorithm~\ref{alg:cont-dg} are also assumed to be done coordinate-wise when applicable. Our presentation differs from that of Ene et al. in two ways. (1) we use a different update rule in the line search (simplifying and removing a logarithmic factor from the round complexity), and (2) we simplify the special cases for when $\nabla_i f(\vx) \leq 0$ or $\nabla_i f(\vx) \geq 0$ with the $\max$ and $\min$ operators. The update rule is derived from the double greedy algorithm of Chen et al.~\cite{CFK19}. 

\paragraph{Discretization and implementation details}
When the line searches are inexact, the analysis of Ene et al. can be applied to show that the errors incurred are at most $O(\gamma OPT)$ in total. Further discretization error is incurred by the termination condition of the while loop. This causes $\vx$ and $\vy$ in the discretized version to not meet exactly at the end of the algorithm. The error for this is at most $\iprod{\nabla f(\vy) - \nabla f(\vx)}{\mathds{1}} \leq \gamma OPT$. Thus at the cost of an $O(\gamma OPT)$ additive error, we can assume in our analysis that the line searches are exact and the points $\vx$ and $\vy$ meet exactly. The algorithm also requires the exact value of $OPT$. This estimate can be obtained to accuracy $(1+o(\gamma))OPT$ via $\log(1/\gamma)$ parallel runs of the algorithm. First, a constant factor approximation of $OPT$ is obtained by sampling a random set. Then, this approximation is multiplied by successive powers of $(1+o(\gamma))$ and the algorithm is run with all guesses in parallel. More details can be found in Ene et al.~\cite{ENV19}. This incurs error at most $o(\gamma OPT)$, so we assume the algorithm knows $OPT$ exactly. 

\begin{algorithm}[ht]
  \DontPrintSemicolon
  Line search for the smallest $\eta_0 \in [0, 1/2)$ such that $\iprod{\nabla f(\eta_0 \mathds{1}) - \nabla f((1 -\eta_0) \mathds{1})}{\mathds{1}} \leq 2 OPT$ \;
  \If{$\eta_0$ does not exist}{
    \KwRet{$f(\frac{1}{2}\mathds{1})$}
  }
  $\vx = \eta_0 \mathds{1}, \vy = (1-\eta_0) \mathds{1}$ \;
  \While{ $\iprod{\nabla f(\vy) - \nabla f(\vx)}{\mathds{1}} \geq \gamma OPT$ }{
    $\mathbf{\Delta}_i^{(x)} = \frac{\nabla_i f(\vx)_+}{\nabla_i f(\vx)_+ - \nabla_i f(\vy)_-}$, $\forall i \in E$ \;
    $\mathbf{\Delta}_i^{(y)} = \frac{\nabla_i f(\vy)_-}{\nabla_i f(\vx)_+ - \nabla_i f(\vy)_-}$, $\forall i \in E$ \;
    
    Line search for the smallest $\eta > 0$ such that  
    \begin{equation*}
    \begin{split}
    \iprod{\nabla f(\vx + \eta \mathbf{\Delta}^{(x)})}{\mathbf{\Delta}^{(x)}} + \iprod{\nabla f(\vy + \eta \mathbf{\Delta}^{(y)})}{\mathbf{\Delta}^{(y)}} \\ 
    \leq \iprod{\nabla f(\vx)}{\mathbf{\Delta}^{(x)}} + \iprod{\nabla f(\vy)}{\mathbf{\Delta}^{(y)}} - \gamma OPT
    \end{split}
    \end{equation*}
    
    $\vx = \vx + \eta \mathbf{\Delta}^{(x)}, \vy = \vy + \eta \mathbf{\Delta}^{(y)}$ \;
  }
  \KwRet{ $\max (f(\vx), f(\vy)) $}  
  \caption{Low adaptivity continuous greedy algorithm. Input is a submodular function $f$ and a parameter $\gamma$. \label{alg:cont-dg}}
\end{algorithm}

Let $\eta_i$ be the step size returned by the line search on iteration $i$ of the while loop and $t_i = \sum_{j=0}^{i-1} \eta_i$ with $t_0 = 0$. Let $\eta_0$ be the result of the line search before the while loop. Through analysis similar to Chen et al.~\cite{CFK19} and Ene et al.~\cite{ENV19}, one can show that Algorithm~\ref{alg:cont-dg} has the following properties:
\begin{enumerate}
    \item The algorithm terminates with a solution $DG$ in $2/\gamma+1$ rounds.
    \item The returned solution $DG$ satisfies \label{eq:dg-bound}
    \begin{equation*}
        DG \geq (1-\gamma/2) \frac{OPT}{2} + \frac{1}{4} \sum_{s>0} \eta_{s} \sum_{i\in E} \frac{ (\nabla_i f(\vx(t_s))_+ + \nabla_i f(\vy(t_s))_-)^2 }{\nabla_i f(\vx(t_s))_+ - \nabla_i f(\vy(t_s))_-}.
    \end{equation*}
\end{enumerate}

For the analysis to be as self-contained as possible, we give proofs of both properties below.

\begin{lemma}
\label{lem:round-analysis}
Algorithm~\ref{alg:cont-dg} terminates in $2/\gamma+1$ rounds.
\end{lemma}

\begin{proof}
Let $\Phi_n =\iprod{\nabla f(\vx(t_n)) - \nabla f(\vy(t_n))}{\mathds{1}}$. We show that each iteration of the while loop decreases $\Phi_n$ by at least $\gamma OPT$.


Since $\mathbf{\Delta}^{(x)} - \mathbf{\Delta}^{(y)} = \mathds{1}$, the line search condition 
can be rewritten as
\begin{equation*}
\begin{split}
& \iprod{\nabla f(\vx + \eta \mathbf{\Delta}^{(x)}) - \nabla f(\vy + \eta \mathbf{\Delta}^{(y)})}{\mathds{1}} \\
& + \iprod{\nabla f(\vy + \eta \mathbf{\Delta}^{(y)})}{\mathbf{\Delta}^{(x)}} + \iprod{\nabla f(\vx + \eta \mathbf{\Delta}^{(x)})}{\mathbf{\Delta}^{(y)}} \leq \\
& \iprod{\nabla f(\vx) - \nabla f(\vy)}{\mathds{1}} + \iprod{\nabla f(\vy)}{\mathbf{\Delta}^{(x)}} + \iprod{\nabla f(\vx)}{\mathbf{\Delta}^{(y)}} - \gamma OPT.
\end{split}
\end{equation*}

By submodularity, $\nabla f(\vx) \geq \nabla f(\vx + \eta \mathbf{\Delta}^{(x)})$. Since $\mathbf{\Delta}^{(y)} \leq 0$, this implies that $\iprod{\nabla f(\vx + \eta \mathbf{\Delta}^{(x)})}{\mathbf{\Delta}^{(y)}} \geq \iprod{\nabla f(\vx)}{\mathbf{\Delta}^{(y)}}$. 
Similarly, $\iprod{\nabla f(\vy + \eta \mathbf{\Delta}^{(y)})}{\mathbf{\Delta}^{(x)}} \geq \iprod{\nabla f(\vy)}{\mathbf{\Delta}^{(x)}}$. Thus $\Phi_{n+1} \leq \Phi_n-\gamma OPT$.

Since $\Phi_0 \leq 2OPT$, the algorithm terminates in at most $2/\gamma$ iterations of the while loop (and one iteration of the initial line search).
\end{proof}

Now we prove the second property of Algorithm~\ref{alg:cont-dg}. Let $\vx^*$ be the optimal solution and let $\mathbf{p}(t) = \mathrm{Proj}_{[\vx(t), \vy(t)]}\vx^*$ be the projection of $x^*$ into the box defined by $\vx$ and $\vy$.

\begin{lemma}
\label{lem:apx-usm}
The returned solution $DG$ satisfies $$DG \geq (1-\gamma/2)\frac{OPT}{2} + \frac{1}{4} \sum_s \eta_{s} \sum_{i\in E} \frac{ (\nabla_i f(\vx(t_s))_+ + \nabla_i f(\vy(t_s))_-)^2 }{\nabla_i f(\vx(t_s))_+ - \nabla_i f(\vy(t_s))_-}.$$
\end{lemma}

To show Theorem~\ref{lem:apx-usm}, we borrow the following lemmas of Ene et al. (proof of Lemma 7)\footnote{The analysis actually sets $\dot{x}=\dot{y}=0$ for $i$ such that $\nabla_i f(x) \leq 0$ or $\nabla_i f(y) \geq 0$, thus differing from our definition slightly for these $i$. However, their analysis is easily extended to all of $E$ for our choice of $\dot{x}$ and $\dot{y}$ with little modifications.} and Chen et al. (Corollary 3.19):
\begin{lemma}
With our choice of $\xdot$ and $\ydot$, we have the following inequalities:
$$\int_{0}^{\eta_0} \frac{1}{2}\left( \iprod{\nabla f(\vx)}{\xdot} + \iprod{\nabla f(\vy)}{\ydot} \right) + \iprod{\nabla f(\mathbf{p})}{\dot{\mathbf{p}}} \dt \geq 0,$$
$$\int_{t_i}^{t_i + \eta_i} \iprod{\nabla f(\mathbf{p})}{\dot{\mathbf{p}}} \dt \geq \eta_i \sum_{i\in S(t)} \frac{ \nabla_i f(\vx(t_i)) \nabla_i f(\vy(t_i)) }{\nabla_i f(\vx(t_i)) - \nabla_i f(\vy(t_i))},$$
for all $i \in S(t_i)$ where $S(t) = \{i \in E \, \vert \, \nabla_i f(\vx(t)) \geq 0 \land \nabla_i f(\vy(t)) \leq 0\}$.
\end{lemma}

We are now ready to prove Theorem~\ref{lem:apx-usm}.
\begin{proof}
Let $$\zeta_{i,s} = \iprod{\nabla f(\vx(t_s))}{\mathbf{\Delta}^{(x(t_s))}} + \iprod{\nabla f(\vy(t_s))}{\mathbf{\Delta}^{(y(t_s))}}.$$
We have the following sequence of inequalities:
\begin{equation*}
\begin{split}
    & 2DG - OPT \\
    \, & \geq \int_0^1 \frac{1}{2}\diff{f(\vx) + f(\vy)}{t} + \diff{f(\mathbf{p})}{t} \dt \\
    \, & = \sum_s \int_{t_s}^{t_s + \eta_s} \frac12 \iprod{\nabla f(\vx + \eta \mathbf{\Delta}^{(x)})}{\mathbf{\Delta}^{(x)}} \\
    \, & + \sum_s \int_{t_s}^{t_s + \eta_s}\frac12 \iprod{\nabla f(\vy + \eta \mathbf{\Delta}^{(y)})}{\mathbf{\Delta}^{(y)}} + \iprod{\nabla f(\mathbf{p})}{\dot{\mathbf{p}}} \dt \\
    \, & \geq \sum_{s>0} \eta_s \left( \frac{1}{2}\left(-\gamma OPT + \zeta_{i,s}\right) + \sum_{i\in S(t)} \frac{ \nabla_i f(\vx(t_s)) \nabla_i f(\vy(t_s)) }{\nabla_i f(\vx(t_s)) - \nabla_i f(\vy(t_s))} \right) \\
    \, & \geq -\gamma OPT/2 + \sum_{s>0} \eta_s \sum_{i\in E} \frac{1}{2} \frac{ (\nabla_i f(\vx(t_s))_+ + \nabla_i f(\vy(t_s))_-)^2 }{\nabla_i f(\vx(t_s))_+ - \nabla_i f(\vy(t_s))_-}\\
\end{split}
\end{equation*}
where line 3 comes from the analysis of the line search in Lemma~\ref{lem:round-analysis}, and line 5 comes from completing the square for terms summed over $S(t)$.
\end{proof}

Now we have all the ingredients to prove Theorem~\ref{thm:dg-gains}. To get some intuition for Theorem~\ref{thm:dg-gains}, note that whenever $\nabla f(x)_+ + \nabla f(y)_-  \neq 0$, we get some gain over $OPT/2$ in $DG$ (Property~\ref{eq:dg-bound} of Algorithm~\ref{alg:cont-dg}). More precisely, the gain over $OPT/2$ is proportional to the $L_2$ norm of $\nabla f(x)_+ + \nabla f(y)_-$ under a  certain non-uniform scaling. The norm of the scaling vector can be related again to the performance of the double greedy algorithm. And, we show that the gap between $DG$ and the expected value of a uniformly random subset $R$ is bounded by the $L_1$ norm of $\nabla f(x)_+ + \nabla f(y)_-$. The Cauchy-Schwarz inequality connects these three quantities.

\begin{proof}
Let $RND := \E[f(R)]$ and suppose $RND \leq (1/2 - \delta)OPT$. For the sake of brevity, we omit the $t$ argument on $\vx$ and $\vy$ when it is clear from context.

First we note that if the algorithm terminates on line 3, then $f(\frac{1}{2}\mathds{1})$ has value at least $OPT/2$ and we are done. Suppose for the remainder of the proof that the algorithm has progressed past line 2.

We first estimate the gap between $DG$ and $RND$, by considering the evolution of the point $\frac{\vx(t)+\vy(t)}{2}$ from $\frac12 {\mathds 1}$ to the output of $DG$ (up to an additive $O(\gamma OPT)$):
\begin{equation*}
\label{eq:dg-rnd-gap}
\begin{split}
    DG-RND  & = \int_0^1 \diff{}{t} f\left(\frac{\vx+\vy}{2}\right) \dt \\
            & = \int_0^1 \left\langle \nabla f\left(\frac{\vx(t)+\vy(t)}{2}\right) , \diff{}{t} \left(\frac{\vx(t) + \vy(t)}{2} \right) \right\rangle \dt \\
            & \leq \frac{1}{2} \int_0^1 \sum_{i\in E} \left \lvert \nabla_i f\left(\frac{\vx+\vy}{2}\right) \right \rvert \left \lvert \frac{ \nabla_i f(\vx)_+ + \nabla_i f(\vy)_- }{\nabla_i f(\vx)_+ - \nabla_i f(\vy)_-} \right \rvert \dt\\
            & \leq \frac{1}{2} \int_0^1 \sum_{i\in E} \left \lvert \nabla_i f(\vx)_+ - \nabla_i f(\vy)_- \right \rvert \left \lvert \frac{ \nabla_i f(\vx)_+ + \nabla_i f(\vy)_- }{\nabla_i f(\vx)_+ - \nabla_i f(\vy)_-} \right \rvert \dt \\
             & = \frac{1}{2}  \int_0^1 \sum_{i\in E} \left \lvert \nabla_i f(\vx)_+ + \nabla_i f(\vy)_- \right \rvert \dt
\end{split}
\end{equation*}
where the bound on line 4 is due to the fact that the positive coordinates of $\nabla f\left(\frac{\vx+\vy}{2}\right)$ are upper bounded by $\nabla f(\vx)_+$ and the negative coordinates are lower bounded by $\nabla f(\vy)_-$.

Next we focus on the discretization of the integral. Let 
\begin{equation*}
\begin{split}
\delta_i & = \iprod{\nabla f(\vx) - \nabla f(\vx(t_i) + \eta \mathbf{\Delta}^{(x(t_i))}) }{\mathds{1}} \\
& - \iprod{\nabla f(\vy) - \nabla f(\vy(t_i) + \eta \mathbf{\Delta}^{(y(t_i))})}{\mathds{1}}.
\end{split}
\end{equation*}
For $\eta \in [0, \eta_i)$,
\begin{equation*}
\begin{split}
& \sum_{i\in E} \left \lvert \nabla_i f(\vx(t_i) + \eta \mathbf{\Delta}^{(x(t_i))})_+ + \nabla_i f(\vy(t_i) + \eta \mathbf{\Delta}^{(y(t_i))})_- \right \rvert \\
& - \sum_{i\in E} \left \lvert \nabla_i f(\vx(t_i))_+ + \nabla_i f(\vy(t_i))_- \right \rvert \\
&\leq  \sum_{i\in E} \left(\nabla_i f(\vx(t_i))_+ - \nabla_i f(\vx(t_i) + \eta \mathbf{\Delta}^{(x(t_i))})_+\right) \\
& - \sum_{i\in E} \left(\nabla_i f(\vy(t_i)))_- - \nabla_i f(\vy(t_i) + \eta \mathbf{\Delta}^{(y(t_i))})_-\right) \leq  \delta_i.
\end{split}
\end{equation*}
where line 2 uses the fact that $\nabla f(\vx) \geq \nabla f(\vx + \eta \mathbf{\Delta}^{(x)}) \geq \nabla f(\vy + \eta \mathbf{\Delta}^{(y)}) \geq \nabla f(\vy)$ by submodularity.

From the line search and the proof of Lemma~\ref{lem:round-analysis}, we have for $\eta < \eta_i$:
\begin{equation*}
\begin{split}
\iprod{\nabla f(\vx(t_i) + \eta \mathbf{\Delta}^{(x(t_i))}) - \nabla f(\vy(t_i) + \eta \mathbf{\Delta}^{(y(t_i))})}{\mathds{1}} \\
\geq \iprod{\nabla f(\vx(t_i)) - \nabla f(\vy(t_i))}{\mathds{1}} - \gamma OPT.
\end{split}
\end{equation*}
Rearranging, we have $\delta_i \leq \gamma OPT.$

Let $\alpha_i(t) = \frac{ ( \nabla_i f(\vx(t))_+ + \nabla_i f(\vy(t))_- )^2 }{\nabla_i f(\vx(t))_+ - \nabla_i f(\vy(t))_-}$ and $\beta_i(t) = \nabla_i f(\vx(t))_+ - \nabla_i f(\vy(t))_-$.
Combining the inequalities from above, we have 
\begin{equation*}
\begin{split}
    & \int_0^1 \sum_{i\in E} \left \lvert \nabla_i f(\vx)_+ + \nabla_i f(\vy)_- \right \rvert \dt \\
    & \leq \sum_{s > 0} \int_{t_s}^{t_s + \eta} \left( \gamma OPT  + \sum_{i\in E} \left \lvert \nabla_i f(\vx(t_s))_+ + \nabla_i f(\vy(t_s))_- \right \rvert \right) \dt \\
    & \leq \gamma OPT + \sum_{s > 0} \eta_s \sum_{i\in E} \left \lvert \nabla_i f(\vx(t_s))_+ + \nabla_i f(\vy(t_s))_- \right \rvert \\
    & \leq \gamma OPT + \sqrt{ \sum_{s > 0, i\in E} \eta_s \alpha_i(t_s) }  \sqrt{\sum_{s>0, i\in E} \eta_s \beta_i(t_s) }.
\end{split}
\end{equation*}
where the third line comes from applying Cauchy-Schwarz in the Euclidean norm (over the combined sum over $s$ and $i$).

Next, we show a bound relating $DG$ and the scaling vector $\nabla f(\vx)_+ - \nabla f(\vy)_-$. 
\begin{equation*}
\begin{split}
    DG 
       & \geq \int_0^1 \frac{1}{2} \left(\iprod{\nabla f(\vx)}{\xdot} + \iprod{\nabla f(\vy)}{\ydot}\right)  \dt \\
       & = \frac{1}{4} \int_0^1 \sum_{i\in E} \left( \alpha_i(t) + \beta_i(t)\right)\dt \\ 
       & \geq \frac{1}{4} \sum_{s>0} \int_{t_s}^{t_s+\eta_s} \sum_{i\in E} \beta_i(t) \dt \\
       & \geq \frac{1}{4} \sum_{s>0} \int_{t_s}^{t_s+\eta_s} \left( -\gamma OPT + \sum_{i\in E} \beta_i(t_s) \right) \dt \\
       & \geq -\frac{\gamma}{4}OPT + \frac{1}{4} \sum_{s>0} \eta_s \sum_{i\in E} \beta_i(t_s)
\end{split}
\end{equation*}
Consequently, $$\sum_{s>0, i\in E} \eta_s \beta_i(t_s) \leq (4+\gamma)DG \leq (4+\gamma)OPT.$$
Combining everything together, we have
\begin{equation*}
\begin{split}
DG-RND & \leq \frac{1}{2}  \int_0^1 \sum_{i\in E} \left \lvert \nabla_i f(\vx)_+ + \nabla_i f(\vy)_- \right \rvert \dt \\
       & \leq \frac{\gamma}{2}OPT + \frac{1}{2} \sqrt{ \sum_{s > 0, i\in E} \eta_s \alpha_i(t_s) }  \sqrt{\sum_{s>0, i\in E} \eta_s \beta_i(t_s) } \\
       & \leq \frac{\gamma}{2}OPT + \frac{1}{2} \sqrt{4DG - 2(1-\gamma/2)OPT}\sqrt{(4+\gamma)OPT}.
\end{split}
\end{equation*}

Thus if $RND= (1/2 - \delta)OPT$, then $$(DG - RND)^2 \leq (4+\gamma)(DG - (1-\gamma/2)OPT/2)OPT.$$ This inequality is quadratic in $DG$, and solving for a lower bound on $DG$ yields $DG = (1/2+\Omega(\delta^2)-O(\gamma))OPT.$

Finally, the result on the number of rounds follows from setting $\gamma = O(\delta^2)$.
\end{proof}

\appendix

\section{Some basics on submodular functions}
\label{app:basics}

We use the following facts in our constructions of instances of submodular maximization.
These properties have been used in previous work, see e.g. \cite{V09}.

\begin{lemma}
\label{lem:discretize}
Suppose that $F:\RR_+^n \rightarrow \RR_+$ is a continuous function, such that 
\begin{itemize}
    \item for every $\bx \in \RR_+^n$ and $i \in [n]$, the partial derivative $\partdiff{F}{x_i}\Big|_{\bx + t \be_i}$ is defined and continuous almost everywhere as a function of $t \in \RR_+$,
    \item the partial derivative $\partdiff{F}{x_i}$ (wherever defined) is non-negative and non-increasing in all coordinates.
\end{itemize}

Let $k \geq 1$ and $[n] = E_1 \cup \ldots \cup E_n$ be any partition of the ground set. Then
$$ f(S) = F\left(\frac{|S \cap E_1|}{k}, \frac{|S \cap E_2|}{k}, \ldots, \frac{|S \cap E_n|}{k} \right) $$
is a monotone submodular function.
\end{lemma}

\begin{proof}[Sketch of proof]
If we denote $x_i = \frac{1}{k} |S \cap X_i|$, the marginal values of $f$ are 
\begin{eqnarray*}
f(S+i) - f(S) & = & \int_{0}^{1/k} \partdiff{F}{x_i} \Big|_{\bx+t \be_i} dt
\end{eqnarray*}
considering that $\partdiff{F}{x_i}$ is defined almost everywhere as a function of $t$.
The non-negativity of $\partdiff{F}{x_i}$ implies that $f$ in monotone, and the non-increasing property of $\partdiff{F}{x_i}$ implies that $f$ is submodular.
\end{proof}

\begin{lemma}
\label{lem:monotone-combine}
If $F_1:\RR_+^n \rightarrow [0,1]$ and $F_2:\RR_+^n \rightarrow [0,1]$ satisfy the assumptions of Lemma~\ref{lem:discretize}, then so does $F:\RR_+^n \rightarrow [0,1]$,
$$F(\bx) = 1 - (1-F_1(\bx))(1-F_2(\bx)).$$
\end{lemma}

\begin{proof}[Sketch of proof]
Assuming that the values of $F_i$ are in $[0,1]$, the same holds true for $F$.
The partial derivatives of $F$ are
$$ \partdiff{F}{x_i} = (1-F_1(\bx)) \partdiff{F_2}{x_i} + (1-F_2(\bx)) \partdiff{F_1}{x_i}.$$
Assuming that $\partdiff{F_i}{x_i}$ are non-negative and non-increasing, and hence $F_1$ and $F_2$ are non-decreasing, $\partdiff{F}{x_i}$ is non-negative and non-increasing as well.
\end{proof}
\section{Hard instance for \texorpdfstring{$1-1/e$}{1-1/e}}
\label{sec:1-1/e}

Here we prove Lemma~\ref{lemma:g-properties}, which describes a function implying the hardness of achieving any factor better than $1-1/e$ (which is embedded in the set $Y$ as discussed above). Although such instances are well known by now, we need a new variant, which proves the hardness of achieving a factor better than $1 - 1/e + 1/n^c$ for some constant $c>0$, as opposed to $1 - 1/e + 1/\log^c n$, which would follow for example from \cite{V13}. Also, we need the property that in the symmetric region it takes the form $g(y_1,\ldots,y_{\ell'}) = 1 - e^{-\sum y_j}$ (as opposed to $g(y_1,\ldots,y_{\ell'}) = 1 - (1-\frac{1}{\ell'} \sum y_j)^{\ell'}$). We provide a self-contained construction here.

\begin{lemma}
\label{lemma:1-1/e}
For any $0<\eps<1$ and $\ell' \geq 2/\eps^2$, there is a function $g:[0,1]^{\ell'} \rightarrow [0,1]$ such that
\begin{itemize}
    \item $g$ is continuous, non-decreasing in each coordinate, and it has first partial derivatives $\partdiff{g}{y_i}$ almost everywhere which are non-increasing in every coordinate.
    \item We have $g(1,0,0,\ldots,0) = 1-\eps$.
    \item If $|y_i- \frac{1}{\ell'} \sum_{i=1}^{\ell'} y_i| \leq \eps/2$ for all $i \in [\ell']$,
     then $g(\by) = \min \{ 1 - e^{-\sum_{j=1}^{\ell'} y_j},  1-\eps \}.$
\end{itemize}
\end{lemma}

\begin{proof}
We start by defining a function of a single variable:
\begin{itemize}
    \item For $x \in [0,\eps]$, $\gamma(x) = 1 - e^{-x}$.
    \item For $x \in [\eps,1]$, $\gamma(x) = 1 - e^{-\eps} (1-x+\eps)$.
\end{itemize}
This function is continuous ($\gamma(\eps) = 1 - e^{-\eps}$ according to both definitions), non-decreasing, and its derivative is continuous and non-increasing ($\gamma'(x) = e^{-x}$ for $x \in [0,\eps]$, and $\gamma'(x) = e^{-\eps}$ for $x \in [\eps,1]$). For $x \in [0,1]$, $\gamma(x) \in [0,1]$.

Then we define
$$ g(y_1,\ldots,y_{\ell'}) = \min \left\{ 1 - \prod_{i=1}^{\ell'} (1 - \gamma(y_i)),  1 - \eps \right\}.$$

Clearly, this is a continuous non-decreasing function. As long as $\prod_{i=1}^{\ell'} (1 - \gamma(y_i)) > \eps$, the partial derivatives are
$$ \partdiff{g}{y_i} = \gamma'(y_i) \prod_{j \neq i} (1 - \gamma(y_j)) $$
which is non-increasing in each coordinate $y_j$ (since $\gamma'(y_i)$ is non-increasing and $\gamma(y_j)$ is non-decreasing). For $\prod_{i=1}^{\ell'} (1 - \gamma(y_i)) < \eps$, the partial derivatives are $0$. ($\partdiff{g}{y_i}$ is discontinuous at $\prod_{i=1}^{\ell'} (1 - \gamma(y_i)) = \eps$ but that is at most one point $t$ on any line $\by(t) = \bx + t \be_i$).)

Consider $y_1=1$ and $y_2=\ldots=y_n=0$. We have $\gamma(0) = 0$ and $\gamma(1) = 1 - e^{-\eps} \eps$.
Therefore $g(1,0,\ldots,0) = \min \{ 1 - e^{-\eps} \eps, 1-\eps \} = 1-\eps$.

Finally, let $\bar{y} = \frac{1}{\ell'} \sum_{i=1}^{\ell'} y_i$ and suppose that $|y_i- \bar{y}| \leq \eps/2$ for all $i \in [\ell']$. Then we distinguish two cases:
\begin{itemize}
    \item If $ \bar{y} \leq \eps/2$, then $y_i \leq \eps$ for all $i$. Therefore, $\gamma(y_i) = 1 - e^{-y_i}$ and we obtain
    $$ g(y_1,\ldots,y_{\ell'}) = \min \left\{ 1 - e^{-\sum_{i=1}^{\ell'} y_i}, 1-\eps \right\}.$$
     \item If $\bar{y} > \eps/2$, then $g(\by)$ must be close to $1$: by the AMGM inequality,
     \begin{equation*}
     \begin{split}
     & 1 - \prod_{i=1}^{\ell'} (1-\gamma(y_i)) \geq 1 - \left( \frac{1}{\ell'} \sum_{i=1}^{\ell'} (1-\gamma(y_i)) \right)^{\ell'} \\ 
     & \geq 1 - \left( \frac{1}{\ell'} \sum_{i=1}^{\ell'} (1 - y_i (1 - \eps)) \right)^{\ell'} = 1 - \left( 1 - \bar{y} (1-\eps) \right)^{\ell'}
     \end{split}
     \end{equation*}
     where we used the fact that $\gamma(y_i) \geq y_i (1-\eps)$ for all $y_i \in [0,1]$.
     Hence, using the assumptions that $\bar{y} > \eps/2$ and $\ell' \geq 2/\eps^2$,
     $$ 1 - \prod_{i=1}^{\ell'} (1-\gamma(y_i)) \geq 1 - e^{-\ell' \bar{y} (1-\eps)} \geq 1 - e^{-(1-\eps) / \eps} \geq 1 - \eps $$
     where we used $e^{-t} \leq \frac{1}{1+t}$ with $t=\frac{1-\eps}{\eps}$. Consequently, $g(y_1,\ldots,y_{\ell'}) = 1-\eps$. \qedhere
\end{itemize}
\end{proof}

As a corollary, we present the following implication for the submodular maximization problem. We are not aware of a prior hardness result showing a hardness factor better than $1-1/e+1/\log^c n$.

\begin{theorem}
\label{thm:1-1/e}
For the monotone submodular maximization problem subject to a cardinality constraint, $\max \{f(S): |S| \leq k \}$, any $(1-1/e+\Omega(1/n^{1/4}))$-approximation algorithm on instances with $n$ elements would require exponentially many value queries.
\end{theorem}

\begin{proof}
Consider $n = 4^\ell$, $k = \ell' = \sqrt{n} = 2^\ell$ and $\eps = 2/n^{1/4}$.
Let $[n] = Y_1 \cup Y_2 \cup \ldots \cup Y_{\ell'}$ be a uniformly random partition of $[n]$ into $\ell' = \sqrt{n}$ blocks of size $k = \sqrt{n}$. We consider a monotone submodular function based on Lemma~\ref{lemma:1-1/e},
$$ f(S) = g\left( \frac{|S \cap Y_1|}{k}, \frac{|S \cap Y_2|}{k}, \ldots, \frac{|S \cap Y_{\ell'}|}{k} \right).$$ 
The optimization problem $\max \{f(S): |S| \leq k\}$ has the solution $S = Y_1$ (or any other $Y_i$), which has value
$$ f(Y_1) = g(1,0,\ldots,0) = 1-\eps $$
by Lemma~\ref{lemma:1-1/e}. 

We claim that an algorithm cannot find a solution of value better than $1-1/e$, by arguments which are quite standard by now. For any fixed query $Q$, the fractions $y_i = \frac{|Q \cap Y_i|}{k}$ are well concentrated around their expectation which is $\bar{y} = \frac{1}{k \ell'} |Q| = \frac{1}{n} |Q|$. $Q \cap Y_i$ is a binomial random variable in the range $\{0,\ldots,\sqrt{n}\}$ and hence
by Chernoff-Hoeffding bounds, the standard deviation for $Q \cap Y_i$ is $O(n^{1/4})$ and hence $|y_i - \bar{y}| > \frac{1}{n^{1/4}} = \eps/2$ with exponentially small probability. Unless the algorithm issues exponentially many queries, with high probability it will never learn any information about the partition $(Y_1,\ldots,Y_{\ell'})$ and it will return a solution which again satisfies $|y_i - \bar{y}| \leq \eps/2$ for all $i$ with high probability. The value of any such solution, under the constraint that $\sum y_i \leq 1$, is at most $1-1/e$ by Lemma~\ref{lemma:1-1/e}.
\end{proof}

\bibliographystyle{ACM-Reference-Format}
\bibliography{paper/refs}

\end{document}